\documentclass[12pt]{amsart}

\usepackage{amsmath, amsthm, amssymb, cite}
\usepackage{graphicx}
\usepackage{enumerate}

\theoremstyle{plain}

\newtheorem{theorem}{Theorem}
\newtheorem*{theorem*}{Theorem}
\newtheorem{lemma}[theorem]{Lemma}

\def\tr{\mbox{trace}}

\title{On the quantum graph spectra of graphyne nanotubes}
\author{Ngoc T. Do}
\date{}

\begin{document}
\maketitle{}
\begin{abstract}
We describe explicitly the dispersion relations and spectra of periodic Schr\"odinger operators on a graphyne nanotube structure. 
\end{abstract}

\section{Introduction}
Discovery and study of carbon nanotubes (see, e.g., \cite{Harris, Katsnelson}) predate significantly the appearance of graphene. 
However, logically, the nanotubes can be understood as sheets of graphene (see, e.g., \cite{Katsnelson}) rolled onto a cylinder (i.e., with one of the vectors of the lattice of periods quotioned out). 
This allows for deriving spectral properties of all types of nanotubes from the dispersion relation of graphene, as it was done, for instance, for quantum graph models (see \cite[and references therein]{Kuch_Post}).

In the last several years, new $2D$ carbon structures, dubbed \emph{graphyne(s)}, have been suggested and variety of geometries has been explored (see, e.g. \cite{Ivanovskii,Bardhan,Gorling,Do_Kuch} for details and further references). 
Some of the graphynes promises to have even more interesting properties than graphene (if and when they can be synthesized).

In \cite{Do_Kuch}, one of the simplest (in terms of the fewest number of atoms in the unit cell) graphynes (see Fig. \ref{F:G}) was studied and its complete spectral analysis was done. 
In particular, complete dispersion relation was found and interesting Dirac cones were discovered\footnote{
Near the vertices of these cones the mass of the charge carriers inside the material is effectively zero. 
Thus they travel at an extremely high speed and lead to extraordinary electronic properties (see \cite{Katsnelson}).}. 
It is thus natural to look at the nanotubes obtained by folding a sheet of this particular graphyne. 
This is what we intend to do in this article.
Carbon nanotubes structure and related Hamiltonian are introduced in section \ref{S:nanotube}.
In section \ref{S:main_result} we derive the dispersion relation and band-gap structure for nanotube with main results stated in Theorem \ref{T:main}.
Proofs of supporting lemmas \ref{L:F1}, \ref{L:F2}, \ref{L:F3} are provided in section \ref{S:supporting}. 

\section{Carbon nanotube structures and related Schr\"odinger operators}\label{S:nanotube}

Different analytic and numerical techniques (e.g., tight binding approximation, density functionals, etc.) have been used to study the properties of carbon nano-structures. 
Among them there is one that nowadays is called ``quantum graphs'' (see \cite{Berk_Kuch}).
Namely, one studies Schr\"odinger type operators along the edges of the graph representing the structure. 
Certainly, ``appropriate'' junction conditions at the vertices need to be used.
This type of models has been used in chemistry for quite a while (see, e.g. \cite{Amovilli, Ruedenberg, Ruedenberg_sourcebook} for details and references).

While in \cite{Do_Kuch} we studied spectra of Schr\"odinger operators on the graphyne shown in Fig. \ref{F:G}, now we will introduce carbon nanotubes related to that structure, which will also carry similar Schr\"odinger operators. Studying the spectra of the latter ones is our goal.
\begin{figure}[ht!]
\includegraphics[width=39mm]{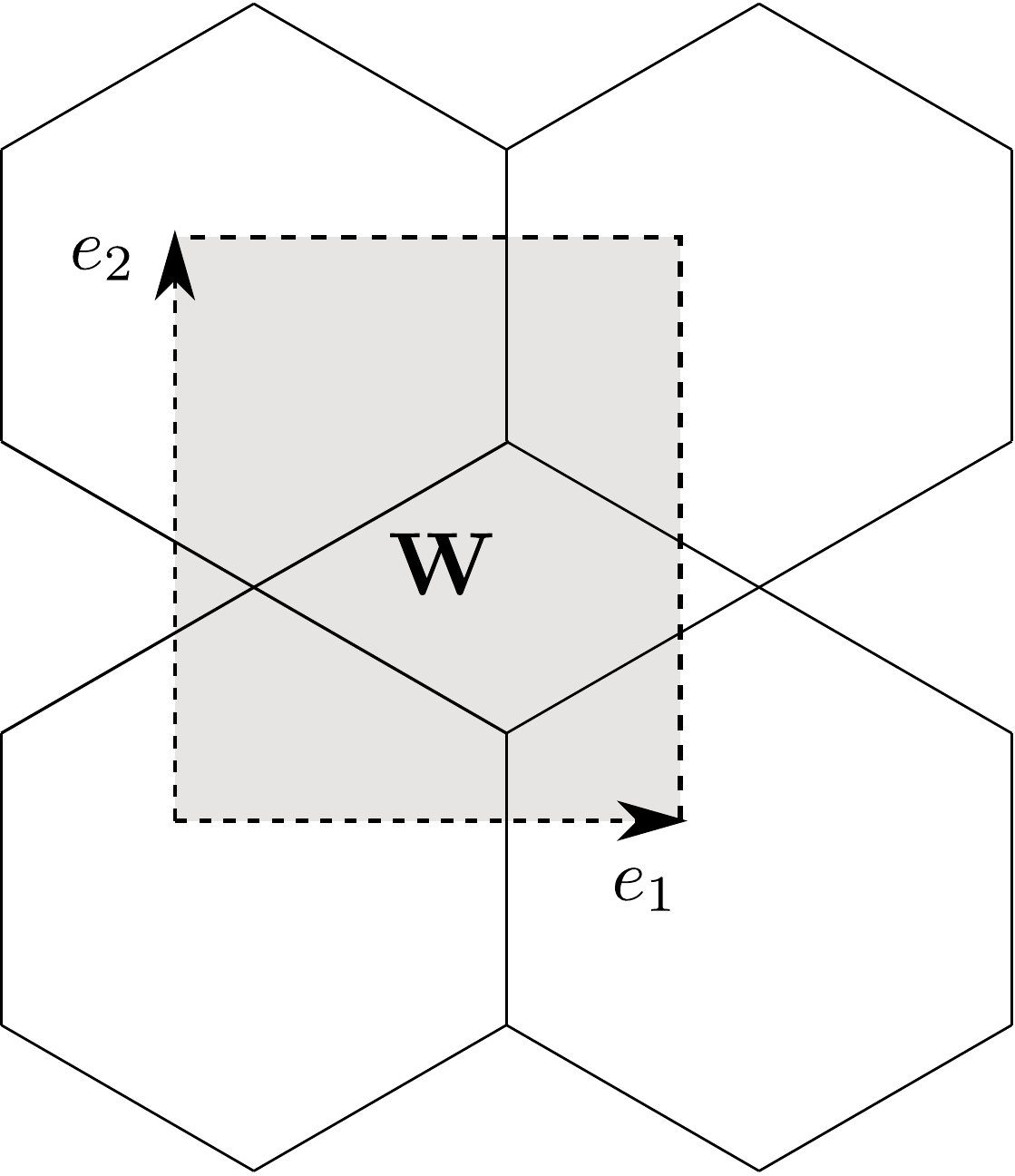}
\caption{Graph $G$ and a fundamental domain $W$}\label{F:G}
\end{figure}
In what is shown in Fig. \ref{F:G}, the vectors $e_1$ and $e_2$ generate the square \emph{lattice} of shifts that leave the geometry invariant. 
The shaded domain $W$ is the unit (fundamental) cell that we choose, which we will call \emph{Wigner-Seitz cell}. 
The lengths of all edges are assumed to be equal to $1$ and each edge is equipped with a coordinate $x$ that identifies it with the segment $[0,1]$.

Let $p=(p_1,p_2)\in \mathbb{Z}^2$\textbackslash$\{(0,0)\}$ be a two-dimensional integer vector, then $pe:=p_1e_1+p_2e_2$ belongs to the latice of translation symmetries of the graphyne $G$.
This means $G+pe=G$.
We define $\iota_p$ to be the equivalence relation that identifies vectors $z_1, z_2\in G$ if $z_2-z_1=kpe$ for some integer $k$.
Then \emph{nanotube} $T_p$ is the graph obtained as the quotient of $G$ with respect to this equivalence relation:
$$
T_p:=G/\iota_p.
$$

Let $q_0(x)$ be a real-valued, even\footnote{The evenness assumption is made not only for technical convenience. 
Without this condition, the graph must be oriented in order to define the operator.
This assumption is also needed to preserve the symmetry that we use.}, square integrable function on [0,1] (i.e., $q_0\in L_2[0,1]$ and $q_0(x)=q_0(1-x)$).
As in \cite{Do_Kuch}, we can use the identifications of the edges with the segment $[0,1]$ to transfer $q_0$ to all edges of $T_p$ and so define a potential $q(x)$ on the whole $T_p$.

In \cite{Do_Kuch}, the operator
\begin{equation}\label{E:graph_op}
   H u(x):=-\frac{d^2u(x)}{dx^2}+q(x)u(x)
\end{equation}
   on the graphyne $G$ was defined and studied. 
Analogously, we introduce the nanotube Schr\"odinger operator $H_p$ that acts on each edge in the same way as $H$ does:
\begin{equation}
   H_p u(x)=-\frac{d^2u(x)}{dx^2}+q(x)u(x)
   \label{E:H}
\end{equation}
and whose domain $D(H_p)$ is the set of all functions $u(x)$ on $T_p$
\footnote{Equivalently, one can say that $u(x)$ is a function on $G$ such that $u(x+kpe)=u(x)$ for all $k\in\mathbb{Z}$.} 
such that:
\begin{enumerate}
 \item $u_e:=u|_e \in H_2(e)$, for all $e \in E(T_p),$
 \item $\sum_{e\in E(T_p)}\|u_e\|^2_{H_2(e)}<\infty$
 \item at each vertex these functions satisfy \emph{Neumann vertex condition},
 i.e. $u_{e_1}(v)=u_{e_2}(v)$ for any edges $e_1,e_2$ containing the vertex $v$ and
 $$
  \sum_{v\in e} u_e' (v)=0, \text{ for any vertex } v \text{ in } V(T_p).
  $$
\end{enumerate}

Understanding the spectra $\sigma(H_p)$ of these nanotube operators is our task here.
\section{Spectra of nanotube operators} \label{S:main_result}
The questions we address here are about the structure of the absolute continuous spectrum $\sigma_{ac}(H_p)$, singular continuous spectrum $\sigma_{sc}(H_p)$, pure point spectrum  $\sigma_{pp}(H_p)$, as well as the shape of the dispersion relation and spectral gaps opening.

In this section, we study the spectra of operator $H_p$ acting on the nanotube $T_p=T_{(p_1,p_2)}$ for some $p=(p_1,p_2)\in\mathbb{Z}^2$. If $p$ is a zero vector, instead of a nanotube one gets the whole graphyne $G$, so we will always assume, without repeating this every time, that $p\neq (0,0)$.

The reciprocal lattice is generated by the $2\pi$-dilations of the pair of vectors $(1/\sqrt{3},0)$, $(0,1/2)$ biorthogonal to the pair $e_1, \ e_2$. 
Using those vectors as the basis, we will use coordinates $(\theta_1,\theta_2)$. In these coordinates, the square $B=[-\pi,\pi]^2$ is a fundamental domain of the reciprocal lattice. Abusing notations again, we will call it the \emph{Brillouin zone} of the graphyne $G$.

The standard Floquet-Bloch theory \cite{Berk_Kuch,Do_Kuch,Eastham,Reed_Simon_4,Kuch_Floquet_book} gives the
following direct integral decomposition of $H$:
\begin{center}
$H=\int\limits_B^\oplus H^\theta d\theta$.
\end{center}
Here $B$ is the Brillouin zone, $H^\theta$ is the Bloch Hamiltonian that acts as (\ref{E:graph_op})
on the domain consisting of functions $u(x)$ that belong to $H^2_{loc}(G)$ and satisfy Neumann vertex conditions and Floquet condition
\begin{equation}\label{E:floquet}
u(x + p_1e_1+ p_2e_2) = u(x)e^{ip\theta}= u(x)e^{i(p_1\theta_1+p_2\theta_2)}
\end{equation}
for all $(p_1,p_2)\in\mathbb{Z}^2$ and all $x\in G$.

Since functions on $T_p$ are in one-to-one correspondence with $p$-periodic functions $u$ on $G$, i.e. $u(x+p_1 e_1+ p_2 e_2)=u(x),$ 
only the values of \emph{quasimomenta} $\theta$ satisfying the condition 
$p\theta=p_1\theta_1+p_2\theta_2\in 2\pi\mathbb{Z}$ will enter the direct integral expansion of $H_p$.
Denoting by $B_p$ the set
\begin{equation}
     B_p=\{(\theta_1,\theta_2)\in B | p\theta=p_1\theta_1+p_2\theta_2=2k\pi, k \in \mathbb{Z}\},
     \label{E:Bp}
\end{equation}
one obtains the direct integral decomposition for $H_p$:
$$
H_p=\int\limits^\oplus_{B_p} H^\theta d\theta.
$$
As a consequence (see, e.g., \cite{Berk_Kuch})
\begin{equation}
     \sigma(H_p)=\bigcup_{\theta\in B_p}\sigma(H^\theta).
     \label{E:restricted_spectrum}
\end{equation}
Moreover, the dispersion relation of $H_p$ is the dispersion relation of $H$ restricted to $B_p$.

We now need to recall some notations and results of \cite{Do_Kuch} that describe the spectrum and dispersion relation of the operator $H$.

We extend the potential $q_0(x)$ on $[0,1]$ to a $1$-periodic function $q_{per}(x)$ on $\mathbb{R}$ and denote by $D(\lambda)=\tr M(\lambda)$ the \emph{discriminant} (or the \emph{Lyapunov function}) of the periodic Sturm-Liouville operator 
$$
H^{per}:=-\frac{d^2}{dx^2}+q_{per}(x).
$$
Here $M(\lambda)$ is the \emph{monodromy} matrix for this operator (see \cite{Eastham})
\footnote{I.e. $M(\lambda)$ is the matrix that transforms the Cauchy data $(u(0)\hspace{2 mm} u'(0))^{\tau}$ of the solution of $H_{per}u=\lambda u$ at zero to the data $(u(1)\hspace{2 mm} u'(1))^{\tau}$ at the end of the period.}.

By $\Sigma^D$ we denote the (discrete) spectrum of the operator $-d^2/dx^2+q_{0}(x)$ on $[0,1]$ with Dirichlet conditions at the ends of this segment.

Finally, we also introduce the triple-valued function 
$ F(\theta):=(F_1(\theta),\\ F_2(\theta),F_3(\theta))$ on $B$ that provides for each $\theta$ the three (real) roots of the equation
\begin{equation}
\qquad 9x^3-x-(\cos\theta_1+1)(3x+ \cos\theta_2)=0.
\label{E:eta_equation}
\end{equation}
We assume here that $F_1\leq F_2\leq F_3$ for each value of $\theta$.

We can now quote the main result of \cite{Do_Kuch}, which describes the spectral structure of the graphyne operator $H$:

\begin{theorem}\cite[Theorem 7]{Do_Kuch}
\indent
\begin{enumerate}
     \item The singular continuous spectrum $\sigma_{sc}(H)$ is empty.

     \item The dispersion relation of operator $H$ consists of the following two parts:\\
     i) pairs $(\theta,\lambda)$ such that $0.5D(\lambda)\in F(\theta)$ (or, $\lambda\in D^{-1}(2F(\theta))$), where $\theta$ is changing in the Brillouin zone;\\
     and \\
     ii) the collection of flat (i.e., $\theta$-independent) branches $(\theta,\lambda)$ such that $\lambda\in\Sigma^D$.

     \item The absolutely continuous spectrum $\sigma_{ac}(H)$ has band-gap structure and is (as the set) the same as the spectrum $\sigma(H^{per})$
     of the Hill operator $H^{per}$ with potential obtained by extending periodically $q_0$ from $[0,1]$.
     In particular, $$\sigma_{ac}(H)=\{\lambda\in\mathbb{R}\big|\left|D(\lambda)\right|\leq 2\},$$
     where $D(\lambda)$ is the discriminant of $H^{per}.$

     \item The bands of $\sigma(H)$ do not overlap (but can touch).
     Each band of $\sigma(H^{per})$ consists of three touching bands of $\sigma(H)$.

     \item The pure point spectrum $\sigma_{pp}(H)$ coincides with $\Sigma^D$ and belongs to the union of the edges of spectral gaps of
     $\sigma(H^{per})=\sigma_{ac}(H)$.
     \\
     Eigenvalues $\lambda\in\Sigma^D$ of the pure point spectrum are of infinite multiplicity and the corresponding eigenspaces are
     generated by simple loop (i.e., supported on a single hexagon or rhombus) states.

     \item Spectrum $\sigma(H)$ has gaps if and only if $\sigma(H^{per})$ has gaps
     \footnote{It is well known that for any non-constant potential $q$, the spectrum $\sigma(H^{per})$ has some gaps (see \cite{Eastham}).
     Moreover, for a generic potential $q$, all gaps are open (see \cite{Simon_generic}).}.
\end{enumerate}
\end{theorem}
Fig. \ref{F:touch} illustrates the statements of \cite[Theorem 7]{Do_Kuch}.
\begin{figure}[ht!]
\includegraphics[width=84mm]{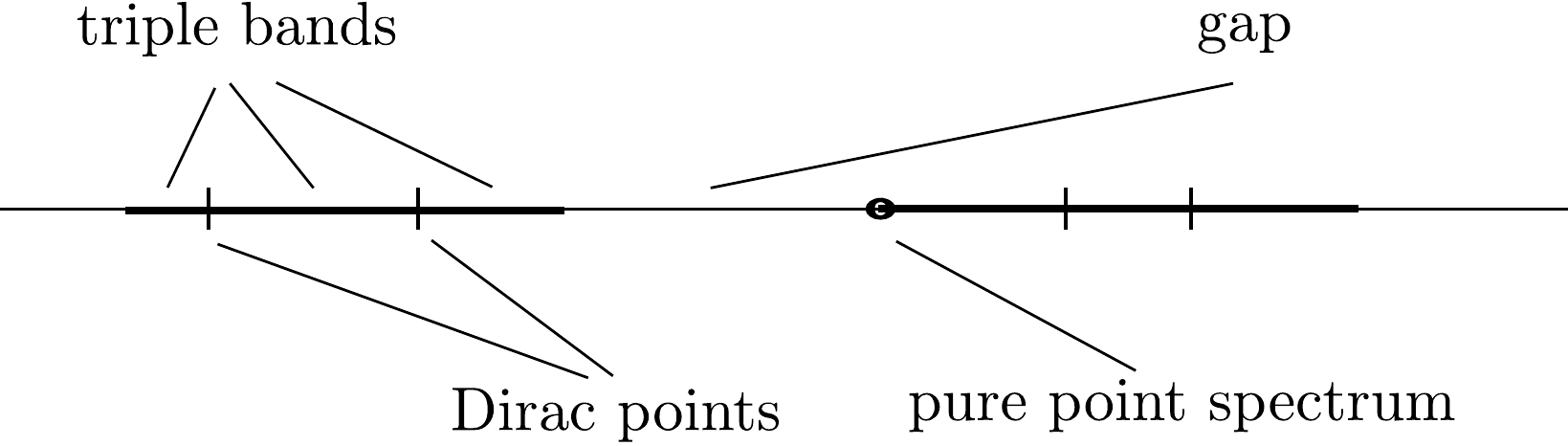}
\caption{Bold segments are bands of $\sigma(H^{per})$. 
Each segment is split into triple bands of $\sigma(H)$ with Dirac points between them. 
One eigenvalue at the end of a band is shown}\label{F:touch}
\end{figure}

Since, in order to obtain the dispersion relation for the nanotube operator $H_p$, we need to restrict this relation to the subset $B_p$ of the Brillouin zone $B$, the previous theorem provides a good start.
Indeed, we see that $\Sigma^D$ belongs to the pure point spectrum $\sigma_{pp}(H_p)$
and the rest of the spectrum is defined by $D^{-1}(2F(B_p))$. 
However, further analysis is still needed, since during the restriction to $B_p$ new gaps might open and new bound states might appear. 
These effects are also expected to depend upon the vector $p$, i.e. on the type of the nanotube (for the ``usual'' nanotubes the names "zig-zag," "armchair," and "chiral" are used, but they are not applicable in our situation).

In what follows, we will study the range of function $F$ restricted to $B_p$.
According to (\ref{E:Bp}), in the $\theta_1\theta_2$-coordinate system $B_p$ is a set of points belonging to a family of parallel lines restricted to the Brillouin zone $B$.
If the slope of these lines is negative then we reflect $B_p$ over the $\theta_2$-axis to make the slope positive. 
Let us denote the new set of points as $R_p$, then $R_p=\{(\theta_1,\theta_2)\in B: |p_2|\theta_2=|p_1|\theta_1-2k\pi, k\in\mathbb{Z}\}$.
Since $F(\theta_1,\theta_2)=F(-\theta_1,\theta_2)$ for all $(\theta_1,\theta_2)$, we have $F(B_p)=F(R_p)$. 
\begin{figure}[ht!]
\includegraphics[width=84mm]{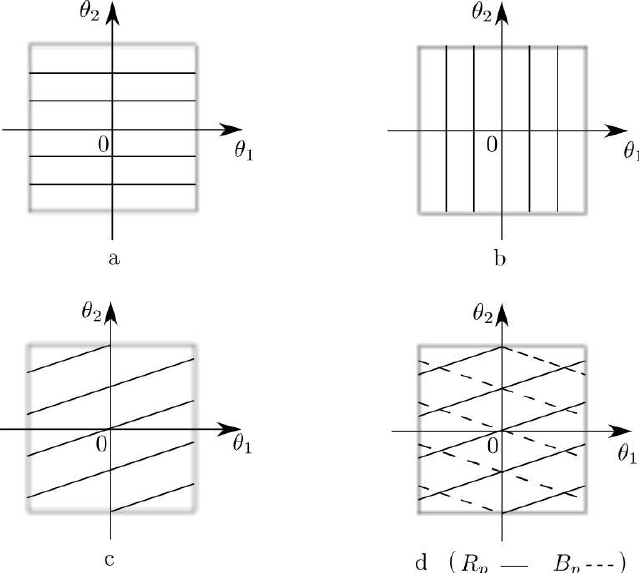}
\caption{All possible cases of $R_p$ vesus $B_p$: a) $p_1=0$, b) $p_2=0$, c) $p_1p_2<0$, d) $p_1p_2>0$}
\end{figure}
We denote by $N_p$ the set of points from $R_p$ that belong to lines with nonnegative $\theta_1$-intercept (or nonpositive $\theta_2$-intercept in case lines from $R_p$ are parallel to $\theta_1$-axis). 
Then $N_p=\{(\theta_1,\theta_2)\in B: |p_2|\theta_2=|p_1|\theta_1-2k\pi, k=0,1,\ldots\}$.
Since $F(\theta_1,\theta_2)=F(-\theta_1,-\theta_2)$ for all $(\theta_1,\theta_2)$, $F(R_p)=F(N_p)$.

Let $V_q:=\{(\theta_1,\theta_2)\in B: q_2\theta_2=q_1\theta_1-2k\pi, k=0,1,\ldots\}$, $q=(q_1, q_2)$.
The above argument proves that $F(B_p)=F(V_q)$ for $q=(q_1,q_2)=(|p_1|,|p_2|)$. 
Thus, it is sufficient to study the range of function $F(\theta)$ restricted to $V_q$ for nonnegative $q_1, q_2$.  

Let us denote $l_0=[q_1\theta_0/2\pi]$. 
Below we state three lemmas about the range of functions $F_j,j=\overline{1,3}$, proofs of which will be provided in Section \ref{S:supporting}.

\begin{lemma}\label{L:F1}
\indent
  If $q_2=0$ and $q_1>1$, then $$F_1(V_q)=[-1,F_1(2l_0\pi/q_1,0)]\cup [F_1(2(l_0+1)\pi/q_1,\pi),-1/3].$$ 
	
	If $q_2=0$ and $q_1=1$, then $F_1(V_q)=[-1,-2/3]$.

  If $q_2\neq 0$ is even, then $F_1(V_q)=\big[-1,-\frac{1}{3}\big]$. 
	
	If $q_2$ is odd, then $F_1(V_q)=[a,-1/3]$ for some $a:=\min F_1(V_q)\in(-1,-2/3]$. In particular, if $q_1=0$ then $a=F_1(0,-2[q_2/2]\pi/q_2)$
\end{lemma}

\begin{lemma}\label{L:F2}
\indent
If $q_1\leq 1$ and $q_2$ is odd, then $F_2(V_q)=[-1/3,a]$ for some $a\in[0,1/3)$. 
Otherwise $F_2(V_q)=[-1/3,1/3]$.
\end{lemma}

\begin{lemma}\label{L:F3}
\indent
If $q_2=0$ and $q_1>1$ then $$F_3(V_q)=[1/3,F_3(2(l_0+1)\pi/q_1,0)]\cup [F_3(2l_0\pi/q_1,\pi),1].$$ 

If $q_2=0$ and $q_1=1$ then $F_3(V_q)=[2/3,1]$.

Otherwise $F_3(V_q)=[1/3,1]$.
\end{lemma}
Taking into account that $F(B_p)=F(V_q)$ where $q=(|p_1|,|p_2|)$, we summarize results of above lemmas in Fig. \ref{F:range}
\begin{figure}[ht!]
\includegraphics[width=84mm]{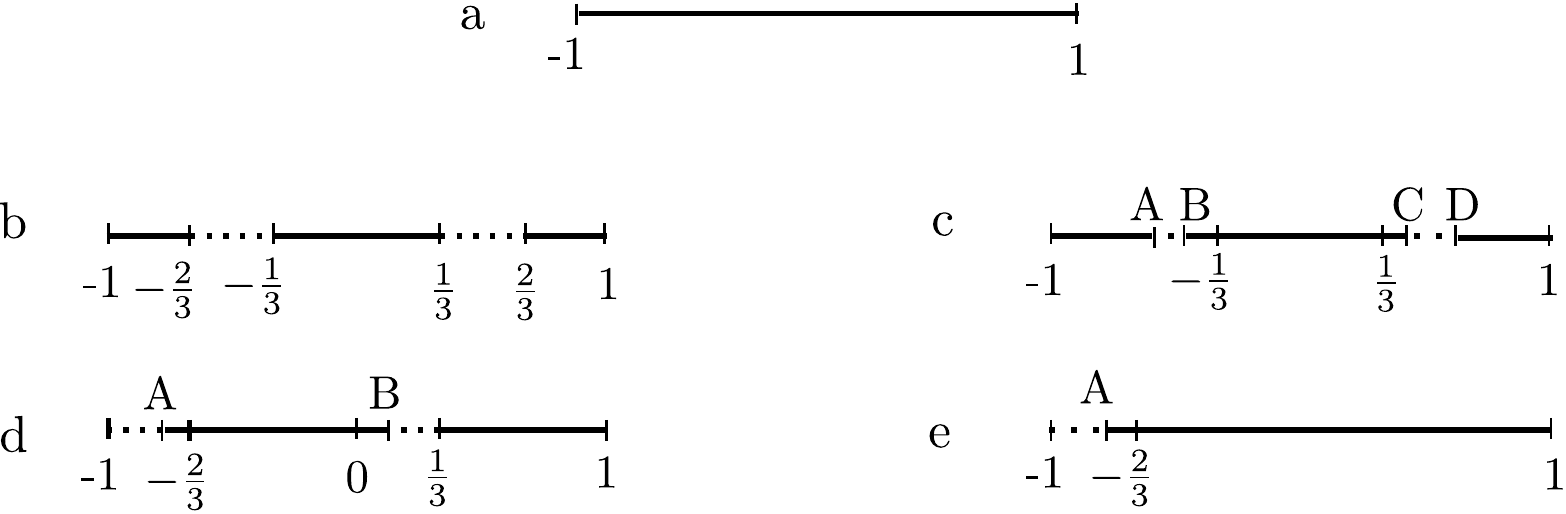} 
\caption{Functions $F_j, j=\overline{1,3}$ have values in $[-1,1]$. 
The bold segments are $F(\theta)$.
The dotted intervals do not belong to the union of the ranges of $F_j, j=\overline{1,3}$.
a) $p_2\neq 0$ even, 
b) $p_2=0$, $p_1=\pm1$,
c) $p_2=0$, $|p_1|>1$, 
d) $p_2$ odd, $|p_1|\leq 1$,
e) $p_2$ odd, $|p_1|>1$} \label{F:range}
\end{figure}

Notice that in case c ($p_2=0$ and $|p_1|>1$) either or both dotted intervals may vanish.

As it was pointed out before, $\Sigma^D$ belongs to the pure point spectrum of $H_p$.
Extra pure point spectrum appears if some non-constant branch of the dispersion relation of $H$ has constant restriction on $B_p$. 
This can happen only on the linear level sets of functions $F_j, j=\overline{1,3}$ inside $B$. 
In \cite[Proposition 10]{Do_Kuch} we described all such sets:
\begin{enumerate}[i.]
          \item $A_1=\{(\pm\pi,\theta_2), \theta_2\in [-\pi,\pi]\},$
          \item $A_2=\{(\theta_1,\pm\pi),\theta_1\in [-\pi,\pi]\},$
          \item $A_3=\{(\theta_1,0),\theta_1\in [-\pi,\pi]\},$
          \item $A_4=\{(\theta_1,\pm\pi/2),\theta_1\in [-\pi,\pi]\}$.
\end{enumerate}

Let us now deal with the additional pure point spectrum that arises due to the presence of linear level sets.
At first, we will build compactly supported eigenfunctions corresponding to those extra eigenvalues.
We then prove that these functions generate the whole corresponding eigenspaces.

We look at the first linear level set 
$A_1=\{(\pm\pi, \theta_2), \theta_2\in [-\pi,\pi]\}$.
Let $p=(2N,0)$ for some postivive integer $N$, then $A_1$ can be rewritten as $\{ \theta| p\theta\pm 2N\pi=0\}$.
On this line $F_1(\theta)=-1/3, F_2(\theta)=0$ and $F_3(\theta)=1/3$.

Let us recall that for each $\lambda\notin \Sigma^D$ functions $\varphi_0,\varphi_1$ are two linearly independent solutions of the equation
\begin{equation}
-\frac{d^2u}{dx^2}+q_0(x)u=\lambda u
\end{equation}
such that $\varphi_{0,\lambda}(0)=\varphi_{1,\lambda}(1)=1, \varphi_{0,\lambda}(1)=\varphi_{1,\lambda}(0)=0.$
Also, function $\eta(\lambda)$ is defined as following $\eta(\lambda):=\varphi'_{1,\lambda}(1)/\varphi'_{1,\lambda}(0)$.
Then (\cite[Lemma 3]{Do_Kuch}) $\lambda$ is in the spectrum of $H$ if and only if there exists $\theta=(\theta_1,\theta_2)\in B$ such that 
$\eta(\lambda)= F_j(\theta)$ for some $j\in\overline{1,3}$.
We will first consider those values of $\lambda$ such that $\eta(\lambda)=F_2(\theta)=0$, i.e. 
$$\frac{\varphi_{1,\lambda}'(1)}{\varphi_{1,\lambda}'(0)}=0 \text{ or } \varphi_{1,\lambda}'(1)=0.$$
As it was said before, we now build a compactly supported eigenfunction, namely $g(x)$, for $H_{(2N,0)}$ corresponding to these $\lambda$.
On four edges directed toward the vertex $A$ let function $g$ to be equal to $\varphi_1$
while on four edges directed toward $B$ and $C$ (two edges per each vertex) define $g$ be equal to $-\varphi_1$ (see Fig. \ref{F:piece_1}).
One can easily check that the Neumann boundary conditions satisfied at vertex $A$.
\begin{figure}[ht!]
\includegraphics[width=84mm]{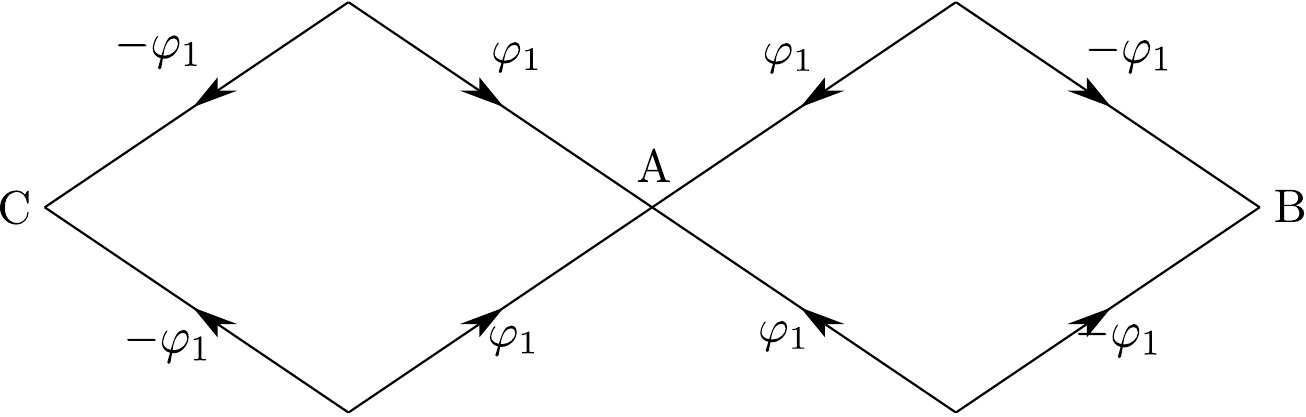}\caption{Piece of rhombus bracelet function corresponding to $p=(2N,0)$ and $\eta(\lambda)=0$}\label{F:piece_1}
\end{figure}

We now extend this piece of function $g$ to the whole nanotube by repeating it $N$ times horizontally.
Outside this band of rhombuses around the nanotube, function $g$ is defined to be equal to zero.
Then $g$ is periodic with period $2e_1$ and satisfies the Neumann boundary conditions at all vertices.
Thus it is a compactly supported eigenfunction for the nanotube $T_{(2N,0)}$ corresponding to those $\lambda$ such that $\eta(\lambda)=0$.
We will call the constructed function as \emph{rhombus bracelet function}.

In Fig. \ref{F:figure_2}, \ref{F:figure_3}, \ref{F:figure_4} one can find similar functions built on a piece of the nanotube structure,
extensions of which will serve as the compactly supported eigenfunctions corresponding to the additional eigenvalues.
\begin{figure}[ht!]
\includegraphics[width=84mm]{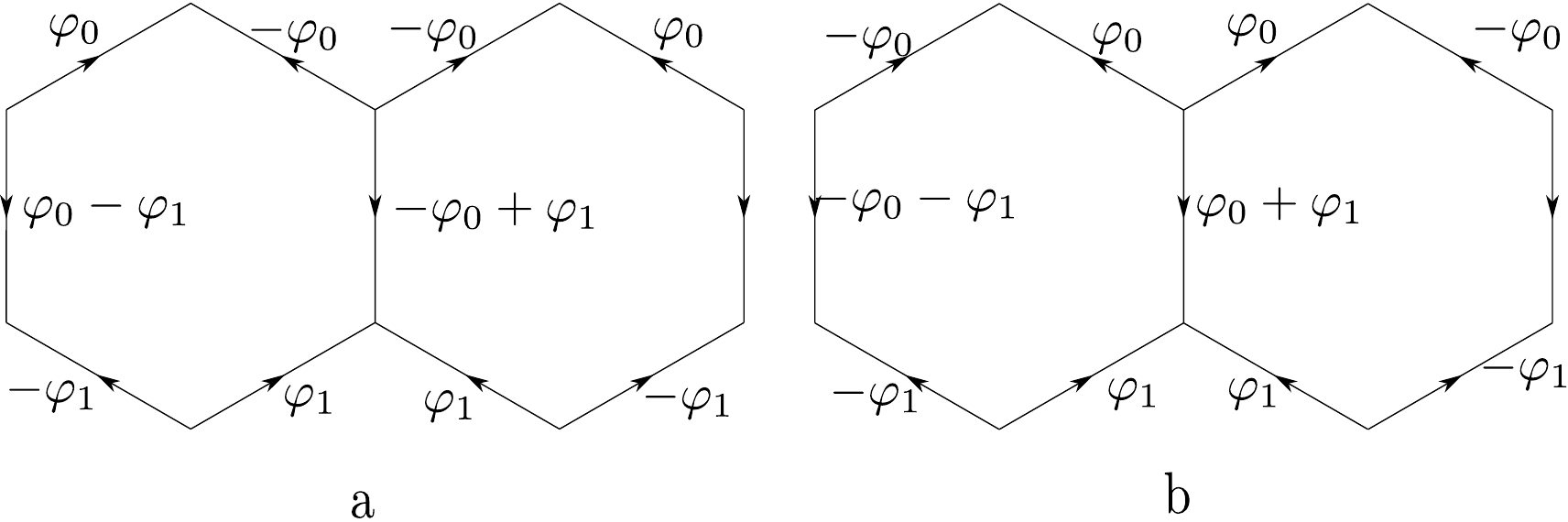}\caption{Piece of hexagon bracelet functions of type a ($\eta(\lambda)=-1/3$) or type b ($\eta(\lambda)=-1/3$) in case $\displaystyle p=(2N,0)$}\label{F:figure_2}
\end{figure}
\begin{figure}[ht!]
\includegraphics[width=84mm]{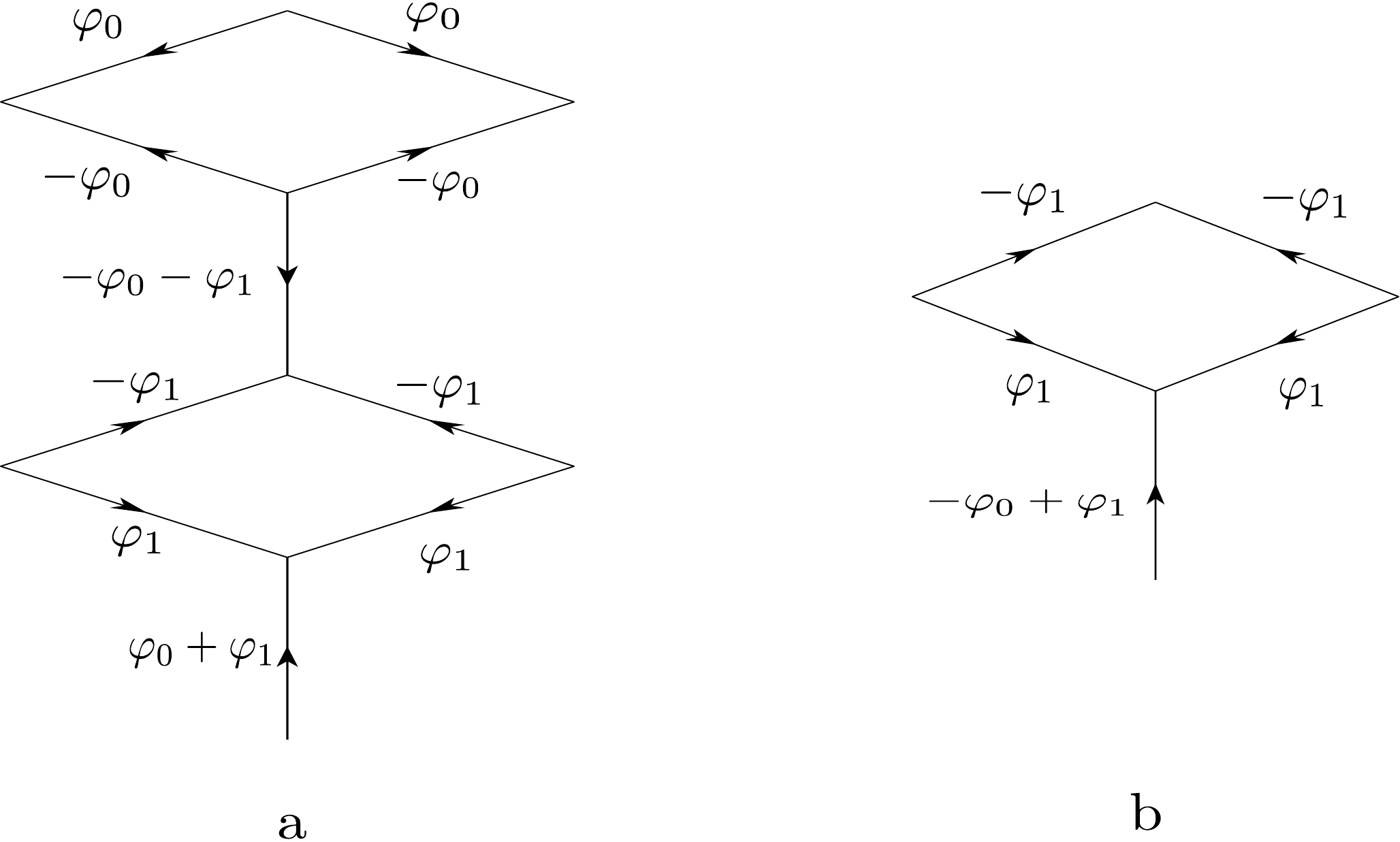}\caption{a. Piece of mushroom function in case $p=(0,N)$ for $N$ - a multiple of $2$ and $\eta(\lambda)=1/3$ and 
                                                 b. piece of flower function in case $p=(0,N)$ for integer $N$ and $\eta(\lambda)=-1/3$}\label{F:figure_3}
\end{figure}
\begin{figure}[ht!]
\includegraphics[width=84mm]{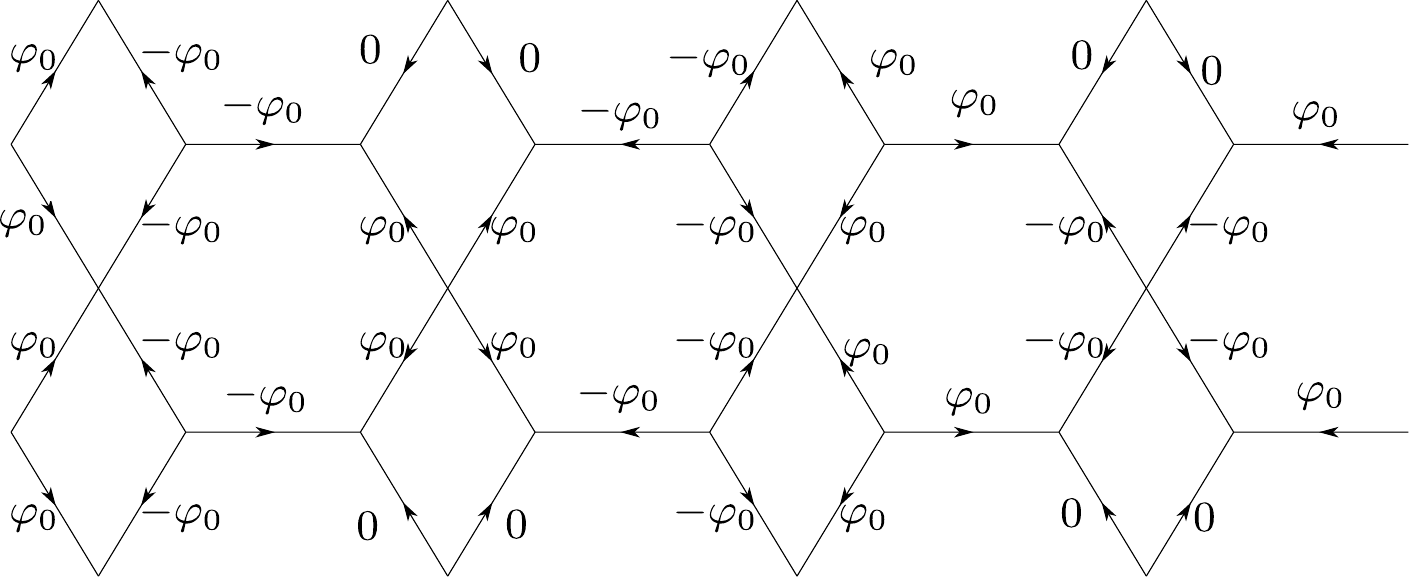}\caption{Piece of double-band function in case $p=(0,N)$ for $N$ - a multiple of $4$ and $\eta(\lambda)=0$}\label{F:figure_4}
\end{figure}
More specifically, in case $p=(2N,0)$ for some nonzero integer $N$, we need to repeat the piece of functions in Fig. \ref{F:figure_2} $N$ times horizontally to obtain a band of hexagons and outside this band functions are defined to be zero.
Then we have compactly supported eigenfunctions corresponding to $\lambda$ such that $\eta(\lambda)=F_1(\theta)=-1/3$ (on the left) and $\eta(\lambda)=F_3(\theta)=1/3$ (on the right).
For the second linear level set $A_2=\{(\theta_1,\pm\pi),\theta_1\in [-\pi,\pi]\}$, we have $p=(0,2N)$ where $N$ is some nonzero integer. 
In this case we has to repeat function from Fig. \ref{F:figure_3}a  $N$ times vertically and beyond that function is defined to be equal to zero.
The obtained function is a compactly supported eigenfunction corresponding to $\lambda$ such that $\eta(\lambda)=1/3$. 
Analogously, $p=(0,N)$ corresponds to the third linear level set $A_3=\{(\theta_1,0),\theta_1\in [-\pi,\pi]\}$. 
One first needs to repeat the piece of function from Fig. \ref{F:figure_3}b $N$ times vertically and then make it equal to zero beyond that in order to get a compactly supported eigenfunction corresponding to $\lambda$ with $\eta(\lambda)=-1/3$. 
The last linear level set, $A_4=\{(\theta_1,\pm\pi/2),\theta_1\in [-\pi,\pi]\}$, corresponds to $p=(0,4N)$ where $N$ is some nonzero integer. 
One again has to repeat the piece of function from Fig. \ref{F:figure_4} $N$ times horizontally and make it equal to zero outside the double-band in order to obtain eigenfunction corresponding to eigenvalues $\lambda$ with $\eta(\lambda)=0$. 
We accordingly call functions in Fig. \ref{F:figure_2} \emph{hexagon bracelet functions of type a and b}, 
in Fig. \ref{F:figure_3} - \emph{mushroom function} and \emph{flower function} accordingly, 
and in Fig. \ref{F:figure_4} - \emph{double-band function}. 

One still needs to prove that these functions generate the whole corresponding eigenspaces. 

Indeed, let $g$ be a compactly supported eigenfunction of $H_{(2N,0)}$ corresponding to those eigenvalues $\lambda$ such that $\eta(\lambda)=0$ and
$P$ - the lowest point on the boundary of the support of $g$ (think about the nanotube $T_{(2N,0)}$ as a vertical tube). 
Since the nanotube structure is periodic, without loss of generality, 
P can be one of three points $D, E$ or $I$ shown in the Fig. \ref{F:figure_1c}.

\begin{figure}[ht!]
\includegraphics[width=39mm]{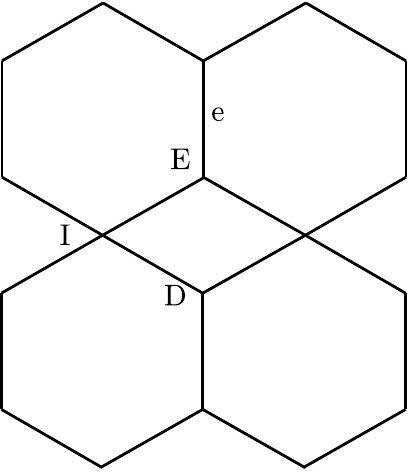}
\caption{Three possible locations ($D, E$ or $I$) of the lowest point on the boundary of the compactly supported eigenfunction in case $p=(2N,0)$}
\label{F:figure_1c}
\end{figure}

The point $P$ cannot coincide with $E$ since it would make $g|_e(E)=g'|_e(E)=0$, and as a result $g|_e\equiv0$ - contradiction. 

Fig. \ref{F:PconcideF} shows what we obtain when trying to construct the compactly supported eigenfunction using Neuman vertex condition in case $P$ coincides with $I$.

\begin{figure}[ht!]
\includegraphics[width=84mm]{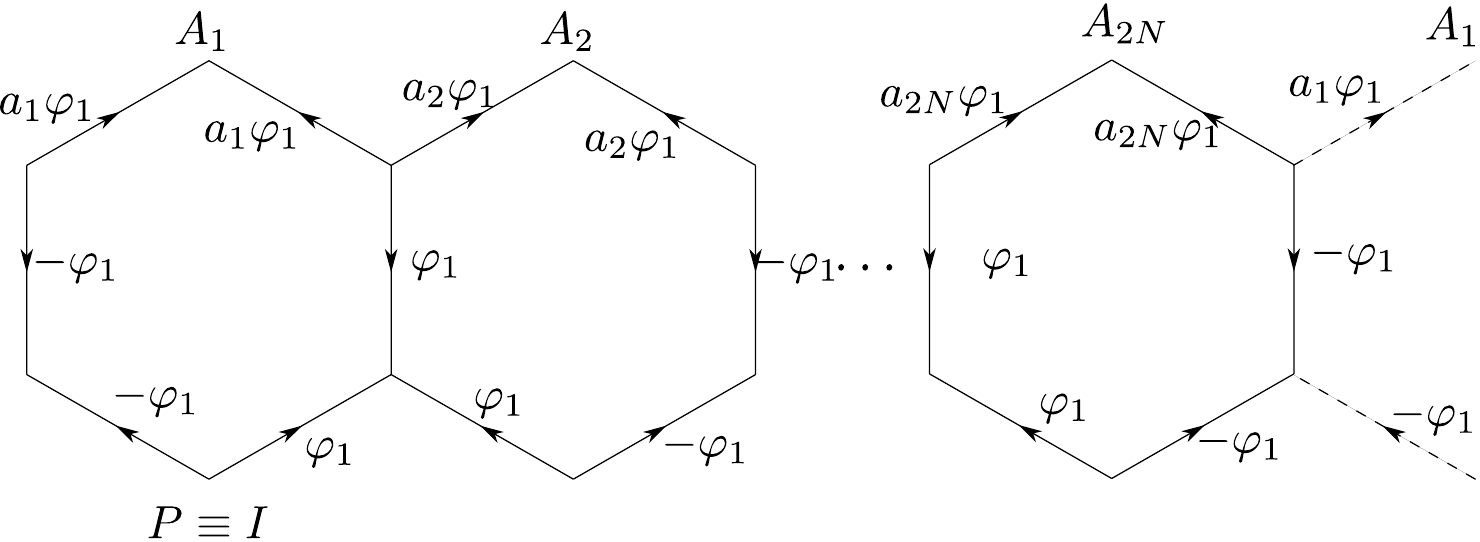}
\caption{Construction of a compactly supported eigenfunction (if such a function exists) when $P$ concides with $I$ in case $p=(2N,0)$ and $\eta(\lambda)=0$}
\label{F:PconcideF}
\end{figure}

From the Neumann boundary condition at vertices $A_j, j=\overline{1,2N}$ we have 
$$a_1+a_2+1=0,$$
$$a_2+a_3-1=0,$$
$$\ldots$$
$$a_{2N-1}+a_{2N}+1=0.$$
$$a_{2N}+a_1-1=0.$$
The sum of $1^{st}, \ldots, {(2N-1)}^{th}$ formulas gives us $a_1+a_2+\ldots+a_{2N}+n=0$ while the sum of $2^{nd}, \ldots, (2N)^{th}$ formulas gives us $a_1+\ldots+a_{2N}-n=0$, 
which lead to contradiction. 
Thus $I$ cannot be the lowest point on the support of function $g$. 

Therefore $D$ must be the lowest point on the boundary of the support of function $g$. 
Extending from $D$ a function with band of rhombuses support and substracting it from $g$, we will get a new function with smaller support. 
Again the lowest point of the new support (provided that the new function is nonzero) must be ``another'' point $D$. 
Continuing this procedure we will eventually get zero function, thus $g$ is a combination of rhombus bracelet eigenfunctions. 

In case $p=(2N,0)$, $N$ - nonzero integer, and $\eta(\lambda)=F_1(\theta)=-1/3$ or $\eta(\lambda)=F_3(\theta)=1/3$ we use the same technique. 
The only difference is the lowest point on the boundary is now $I$. 
Point $E$ is eliminated by the same reason as before and point $D$ is excluded by contradiction obtained from Fig. \ref{F:contradiction}.
\begin{figure}[ht!]
\includegraphics[width=84mm]{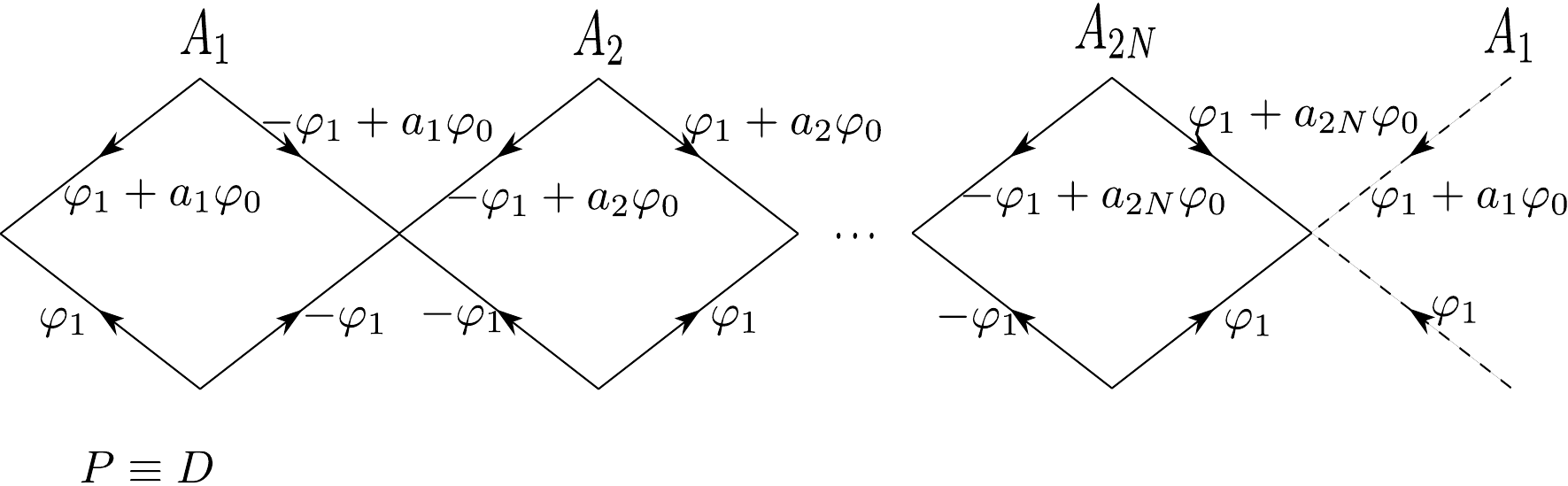}
\caption{Construction of a compactly supported eigenfunction (if such a function exists) when $P$ concides with $D$ in case $p=(2N,0)$ and $\eta(\lambda)=\pm 1/3$}
\label{F:contradiction}
\end{figure} 

In all other cases we also use the ``lowest point'' argument. 
The symmetry of the structure once again follows that the lowest point on the boundary of compactly supported eigenfunction can locate at $D$, $E$ or $I$ (see Fig. \ref{F:lowest_point}). 
\begin{figure}[ht!]
\includegraphics[width=39mm]{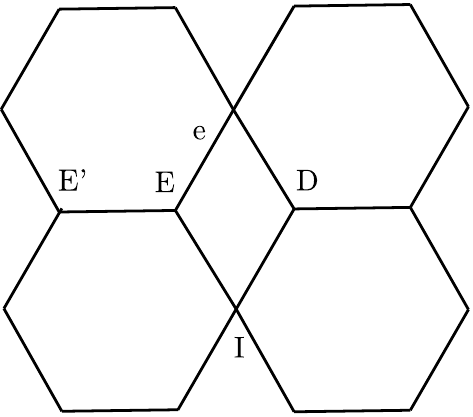}
\caption{Three possible locations of the lowest point on the boundary of the compactly supported eigenfunctions in case $p=(0,N)$ for $N\in\mathbb{Z}$}
\label{F:lowest_point}
\end{figure} 
Both $D$ and $E$ are excluded by the same reason. 
For instance, if $E$ is the lowest point, then $g|_e(E)=g|_e(E')=0$ (otherwise due to the first Neumann boundary condition $E$ will not be the lowest point). 
But then since $g(E)=g(E')=0$ we have $g|_e\equiv 0$, which leads to contradiction. 
The lowest point therefore should be $I$. 

The eliminating process for $p=(0,N), N\in\mathbb{Z}$, $\eta(\lambda)=-1/3$ or $p=(0,2N), N\in\mathbb{Z}$, $\eta(\lambda)=1/3$ occurs exactly the same as in the previous cases.

Now we consider the case when $p=(0,4N), N\in\mathbb{Z}$, $\eta(\lambda)=0$. 
We claim that there does not exist a compactly supported eigenfunction of height $|e_1|$. 
Indeed, suppose the contrary, Fig. \ref{F:double_circle_contradiction} shows what we obtain when constructing such a function. 
But then there does not exist $a$ such that the Neumann boundary conditions satified at both points $A$ and $B$.
\begin{figure}[ht!]
\includegraphics[width=39mm]{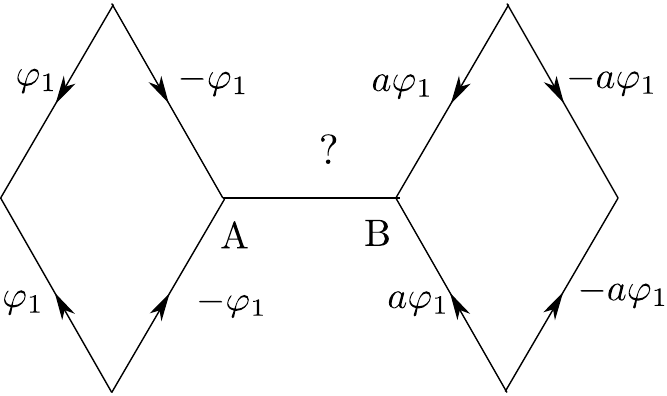}
\caption{Situation when trying to construct a compactly supported eigenfunction of height $|e_1|$}
\label{F:double_circle_contradiction}
\end{figure}
This claim follows that any compactly supported eigenfunction corresponding to $\lambda$ with $\eta(\lambda)=0$ has at least height $2|e_1|$. 

The eliminating process would then be similar to what happens before. 
(Since the minimum height of any nonzero compactly supported eigenfunction is $2|e_1|$, when we substract function of double-band type from the original eigenfunction, there is no need to worry that the support of the new eigenfunction will exceed the old one's.)

Let $\Sigma^0$ be the extra pure point spectrum which occurs due to the linear level set(s) of function $F$. 
Recall that $D(\lambda)=2\eta(\lambda)$.
Then the above argument proves the following:
\begin{lemma}
\indent\label{L:extra_pps}
  \begin{enumerate}
	\item If $p=(2N,0)$ for some nonzero $N\in\mathbb{Z}$ then $\Sigma^0=D^{-1}(\{\pm 2/3,0\})$.
		       
	      Eigenspace corresponding to $\lambda$ with $D(\lambda)=0$ is generated by rhombus bracelet functions. 					
	      Eigenspaces corresponding to $\lambda$ with $D(\lambda)=-2/3$ or $2/3$ are generated by hexagon bracelet functions of type a and b accordingly. 
	\item If $p=(0,N)$ for some odd $N$ then $\Sigma^0=D^{-1}(\{-2/3\})$. 
	      \\
	      If $p=(0,N)$ for some $N$ which is a multiple of $2$ but not a multiple of $4$ then $\Sigma^0=D^{-1}(\{\pm 2/3\})$. 
	      \\
	      If $p=(0,N)$ for some $N$ which is a multiple of $4$ then $\Sigma^0=D^{-1}(\{\pm 2/3,0\})$. 
		        
	      In all cases, the eigenspace corresponding to $\lambda$ with $D(\lambda)=-2/3$ is generated by flower functions.
	      The eigenspace corresponding to $\lambda$ with $D(\lambda)=2/3$ is generated by mushroom functions 
	      and the one which corresponds to $\lambda$ with $D(\lambda)=0$ is generated by double-band functions.
	\end{enumerate}
\end{lemma}	
We are now ready to formulate the main result about the spectra of carbon nanotubes.
Let us first recall some necessary notations.
Vector $pe:=p_1 e_1+ p_2 e_2$ is the translation vector that defines the nanotube $T_p$.
Function $q_0$ is an $L_2$-function on $[0,1]$.
The Hamiltonian $H_p$ is defined on $L_2(T_p)$ with potential $q$ transfered from $q_0$ to each edge.
$B_p$ is the subset of the Brillouin zone $B$ as defined in (\ref{E:Bp}).
The Hill operator $H^{per}$ has potential $q_{per}$, which is periodic extension of $q_0$ and $D(\lambda)$ - its discriminant.
By simple loop state we mean an eigenfunction of operator $H$ with Dirichlet boundary condition whose support is a hexagon or rhombus (see Fig. \ref{F:loopstate}). 
We also introduce so-called tube loop eigenfunction, support of which is a loop of edges around the tube. 
\begin{figure}[ht!]
\includegraphics[width=84mm]{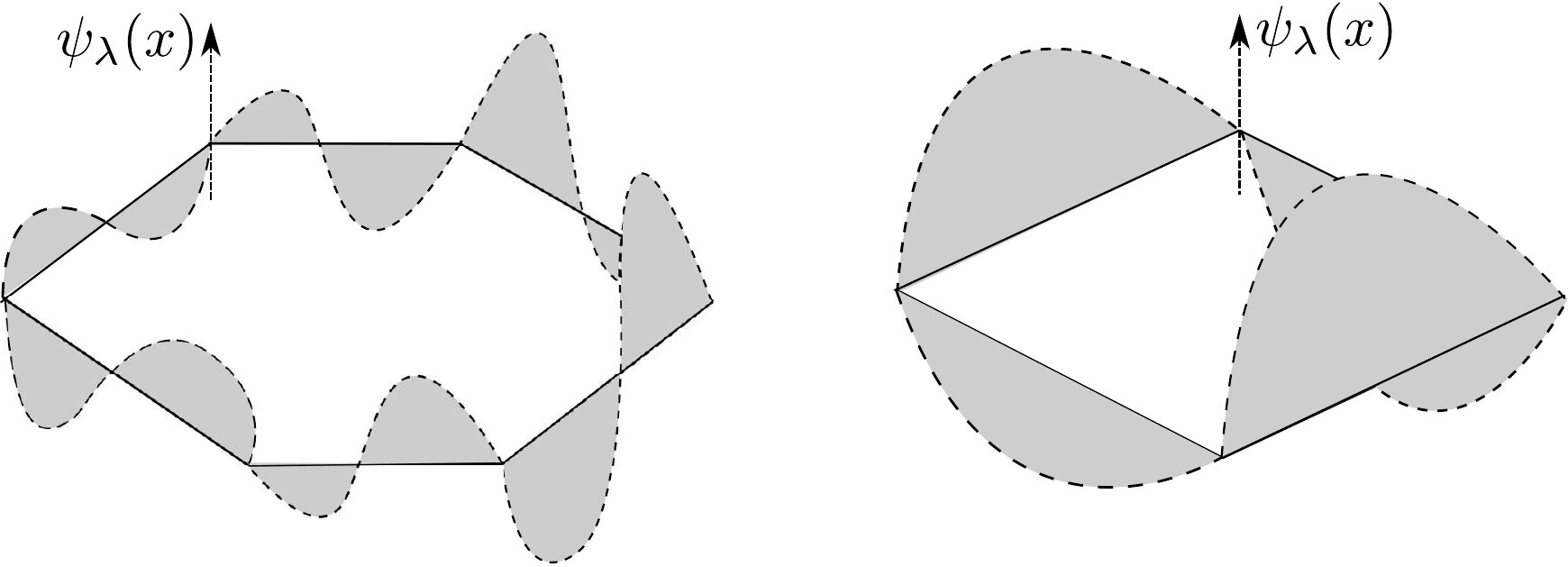}
\caption{Simple loop states constructed from odd function on $[0,1]$ for hexagon and even function on $[0,1]$ for rhombus}\label{F:loopstate}
\end{figure}
\begin{theorem} \label{T:main}
\indent
\begin{enumerate}

    \item The singular continuous spectrum $\sigma_{sc}(H_p)$ is empty.
    
    \item The nonconstant part of the dispersion relation for Hamiltonian $H_p$ is described by the following formula 
    \begin{equation}
    D(\lambda)\in 2F(\theta), \theta\in B_p \text{ or }\lambda\in D^{-1}(2F(\theta)), \theta\in B_p.
    \label{E:dispersion_relation}
    \end{equation}
    
    \item The absolutely continuous spectrum $\sigma_{ac}(H_p)$ has band gap structure. 
		All bands do not overlap. 
		If $p_2$ is nonzero and even then $\sigma_{ac}(H_p)=\sigma(H^{per})$.
    Otherwise there may be additional gaps opened inside spectral bands $H^{per}$.
		In particular
		\begin{enumerate}[i.]
		  \item If $p_2=0$ and $p_1=\pm 1$ then $\sigma_{ac}(H_p)=D^{-1}(\{[-2,-4/3]\cup[-2/3,2/3]\cup[4/3,2]\})$. 
			There is two gaps opened in each spectral band of $H^{per}$.
		
			\item If $p_2=0$ and $p_1\neq \pm 1$ then $\sigma_{ac}(H_p)=D^{-1}(A)$ where 
			$A=[-2,2F_1(2l_0\pi/p_1,0)]\cup[2F_3(2l_0\pi/p_1,\pi),2]\cup$
			\\
			$\cup[2F_1(2(l_0+1)\pi/p_1,\pi),2F_3(2(l_0+1)\pi/p_1,0)].$
			\\
			There are at most two gaps opened in each band of $H^{per}$ depending on whether the following inequalities are true or not
			$F_1(2l_0\pi/p_1,0)<F_1(2(l_0+1)\pi/p_1,\pi)$ \\
			and $F_3(2(l_0+1)\pi/p_1,0)<F_3(2l_0\pi/p_1,\pi)$.
			
		  \item If $p_2$ is odd and $|p_1|\leq 1$ then $\sigma_{ac}(H_p)=D^{-1}(B)$ where
			$B=[2\min F_1(B_p),2\max F_2(B_p)]\cup [2/3,2]$. 
			There is always two gaps opened in each band of $H^{per}$.
			
			\item If $p_2$ is odd and $|p_1|> 1$ then $\sigma_{ac}(H_p)=D^{-1}(C)$ where
			$C=[2\min F_1(B_p),2]$. Only one gap is opened in each band of $H^{per}$. 
		\end{enumerate}
    
    \item The pure point spectrum of $H_p$ contains the pure point spectrum of the Hamiltonian $H$. 
                \begin{enumerate}[i.]
		  \item If $p$ is not of the form $(2N,0)$ or $(0,N)$ for some nonzero integer $N$ then these two sets concide.
	                Eigenvalues from $\Sigma^D$ are of infinite multiplicity and the corresponding eigenspaces are spanned by simple loop state eigenfunctions and tube loop eigenfunctions 
	     	    	   
                  \item If $p=(2N,0)$ or $(0,N)$ for some nonzero $N$, besides $\Sigma^D$ nanotube operator $H_p$ has extra pure point spectrum which is denoted as $\Sigma^0$.
		        All eigenvalues are of infinite multiplicity.
		        Description of these extra eigenvalues and their corresponding eigenspaces are provided in lemma \ref{L:extra_pps}. 
		\end{enumerate}        
\end{enumerate}
\end{theorem}

\begin{proof}
The first claim is a well-known fact about singular continuous spectrum of Schr\"{o}dinger operator (see \cite{Thomas,Berk_Kuch, Reed_Simon_4}).

The second claim follows from the formula (\ref{E:restricted_spectrum}) and claim 2i of \cite[Theorem 7]{Do_Kuch}.

The fact that absolutely continuous spectrum $\sigma_{ac}(H_p)$ has band gap structure and all bands do not overlap is also a consequence of formula (\ref{E:restricted_spectrum}) and claims (3) and (4) of \cite[Theorem 7]{Do_Kuch}. 

Gap will be opened every time when the range of functions $F_1, F_2$ and $F_3$ creates a gap in the interval $[-1,1]$. 
The rest of this claim therefore follows directly from lemmas \ref{L:F1}, \ref{L:F2}, \ref{L:F3} (result of which was briefly described in Fig. \ref{F:range}). 

As it was noticed before, the pure point spectrum of $H_p$ always contains the Dirichlet spectrum $\Sigma^D$. 
The claim about infinite multiplicity of eigenvalues in both cases is known to be true for periodic problems. 
Extra pure point spectrum occurs only if there is some linear level set. 
This can happen when $p=(2N,0)$ or $(0,N)$ for some nonzero $N$. 
The next part of claim (4) is proved similarly as for the analogous one from \cite[Lemma 6]{Do_Kuch}. 
The only difference is the eliminating process may end up with a tube loop eigenfunction.  

The last part of the claim is the same as lemma \ref{L:extra_pps}.
\end{proof}

\section{Proof of supporting lemmas}\label{S:supporting}
In what follows, we denote the graph of the line $q_2\theta_2=q_1\theta_1-2k\pi$ as $t_{q,k}$ and $T_{q,k}$ - part of the line $t_{q,k}$ restricted to $V_q$.
Let us recall that by $\theta_0$ we denote $\arccos(-1/3)$.
\subsection{Proof of Lemma \ref{L:F1}}

First of all we notice the following:
\begin{enumerate}[i.]
  \item For all $\theta_1=\theta_1^0\neq\pm\pi$ fixed, function $F_1(\theta_1^0,\theta_2)$ is decreasing on $[0,\pi]$. 
	\item For all $\theta_2=\theta_2^0$ fixed, function $F_1(\theta_1,\theta_2^0)$ is non-decreasing on $[0,\pi]$.
\end{enumerate}

\begin{figure}[ht!]
\includegraphics [width=84mm]{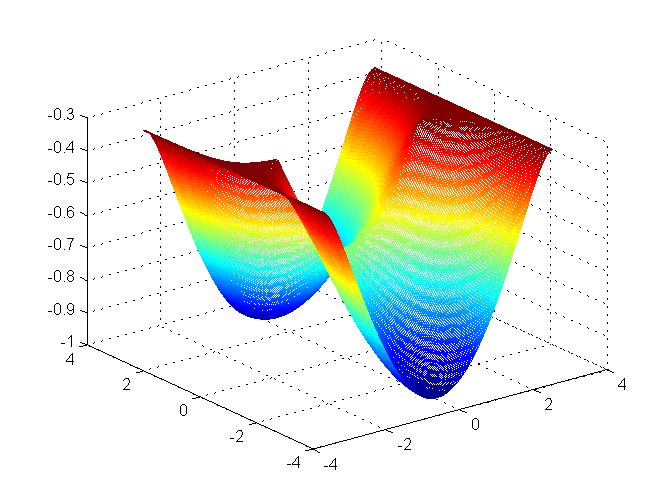}
\caption{Graph of function $F_1$}
\end{figure}
Indeed, $F_1(\theta_1,\theta_2)$ is the smallest value of $x-$ coordinate of intersections of graphs of functions $f(x)=9x^3-x$ and $g(x)=(\cos{\theta_1}+1)(3x+\cos\theta_2)$.
If we fix $\theta_1:=\theta_1^0\neq\pm\pi$, the slope $3(\cos\theta_1^0+1)$ is const and positive.
Then the $y-$intercept $(\cos\theta_1^0+1)\cos\theta_2$ is decreasing on $[0,\pi]$, which makes function $F_1$ decrease on $[0,\pi]$. 
(In case $\theta_1=\pm\pi$, function $F_1$ is constant and equal to $-1/3$.)
\begin{figure}[ht!]
\includegraphics [width=84mm]{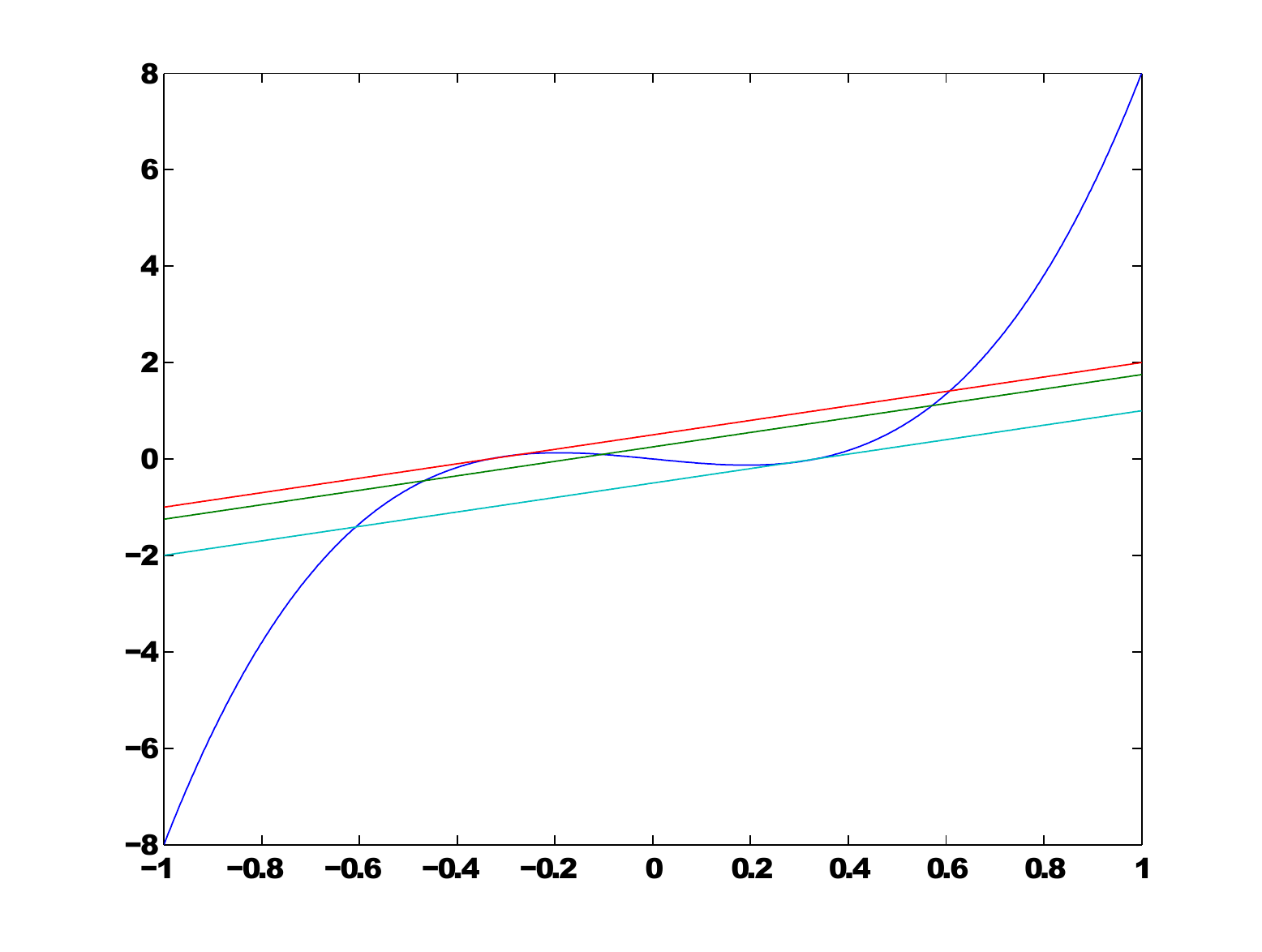}
\caption{The slope is const and $\theta_2$-intercept is increasing make $F_1$ increase}
\end{figure}

Now if we fix $\theta_2:=\theta_2^0$, the $x-$intercept $-(\cos\theta_2^0)/3$ of $g(x)$ will be const and belongs to the interval $[-1/3,1/3]$.
The slope $3(\cos\theta_1+1)$ of function $g(x)$ is decreasing on $[0,\pi]$, thus $F_1(\theta_1,\theta_2^0)$ is nondecreasing on $[0,\pi]$.

If $q_2=0$ then $\displaystyle V_q=\{(\theta_1, \theta_2)\in B: \theta_1=\frac{2k\pi}{q_1}, k=0, 1, 2, \ldots \}$. 

The most right interval $T_{q,k}$ of $V_q$ is $T_{q,[\frac{q_1}{2}]}$. 
Thus, 
$$F_1(V_q)=F_1\big(\bigcup_{0\leq k\leq \big[\frac{q_1}{2}\big]} T_{q,k}\big)=\bigcup_{0\leq k\leq \left[\frac{q_1}{2}\right]} F_1(T_{q,k}).$$

Since $\displaystyle T_{q,k}=\left\{\frac{2k\pi}{q_1}\right\}\times [-\pi,\pi]$ and $F_1(\theta_1, \theta_2)=F_1(\theta_1, -\theta_2)$,
$$F_1(T_{q,k})=F_1\left(\left\{\frac{2k\pi}{q_1}\right\}\times [-\pi,\pi]\right)=F_1\left(\left\{\frac{2k\pi}{q_1}\right\}\times [0,\pi]\right).$$
Now for each $k$, $\theta_1=2k\pi/q_1$ is fixed, according to the remark above we have
$$F_1(T_{q,k})=\left[F_1\left(\frac{2k\pi}{q_1},\pi\right), F_1\left(\frac{2k\pi}{q_1},0 \right)\right].$$
(In case $\theta_1=\pi$ the segment boils down to the one-point set $\{-1/3\}$.)
Therefore,
$$F_1(V_q)=\bigcup_{0\leq k\leq \left[\frac{q_1}{2}\right]} \left[F_1\left(\frac{2k\pi}{q_1},\pi\right), F_1\left(\frac{2k\pi}{q_1},0 \right)\right].$$

Recall that $l_0=[q_1\theta_0/2\pi]$ then the interval $T_{q,k}$ intersects with $\theta_2=\pi$ at $\theta_1=2k\pi/q_1<\theta_0$ for all $k\in[0,l_0]$.
Function $F_1$ is non-decreasing on $[0,\pi]$ for fixed $\theta_2$, thus for all $k\in[0,l_0]$
$$-1=F_1(0,\pi)<F_1(2k\pi/q_1,\pi)\leq -2/3=F_1(\theta_0,\pi)$$
and
$$F_1(0,0)=-2/3\leq F_1(2k\pi/q_1,0) \leq F_1(2l_0\pi/q_1,0).$$
\begin{figure}[ht!]
\includegraphics[width=84mm]{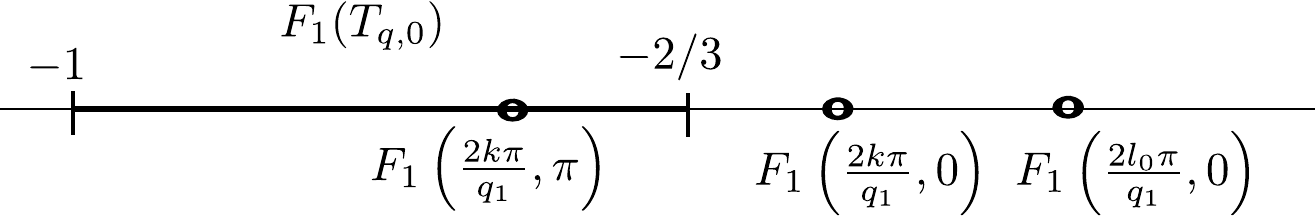}
\caption{Case $q_2=0. F_1(T_{q,0})=[-1,-2/3]$, $F_1(2k\pi/q_1,\pi)$ are on the right of $-2/3$ and $F_1(2k\pi/q_1,0)\leq F_1(2l_0\pi/q_1,0) $ are on the left of $-2/3$ for all $k\in[0,l_0]$}
\label{F:lemma1_vertical_1}
\end{figure}

As a consequence (see Fig. \ref{F:lemma1_vertical_1}),
\begin{equation}
\cup_{0\leq k\leq l_0} F_1(T_{q,k})=[-1,F_1(2l_0\pi/q_1,0)].
\label{E:before_k0}
\end{equation}

For all $q_1\geq 6$ the inequality $1\leq q_1/2-q_1\theta_0/2\pi$ is true, which means there must be an integer $k$ such that $l_0<k\leq [q_1/2]$.
It is not difficult to check that the latter claim is also true for $1< q_1\leq 5$.
Thus, for all $q_1>1$, there exists $k$ such that $T_{q,k}$ belongs to $V_q$ and lies on the right of the line $\theta_1=\theta_0$ (see Fig. \ref{F:lemma1_10}).
\begin{figure}[ht!]
\includegraphics [width=39mm]{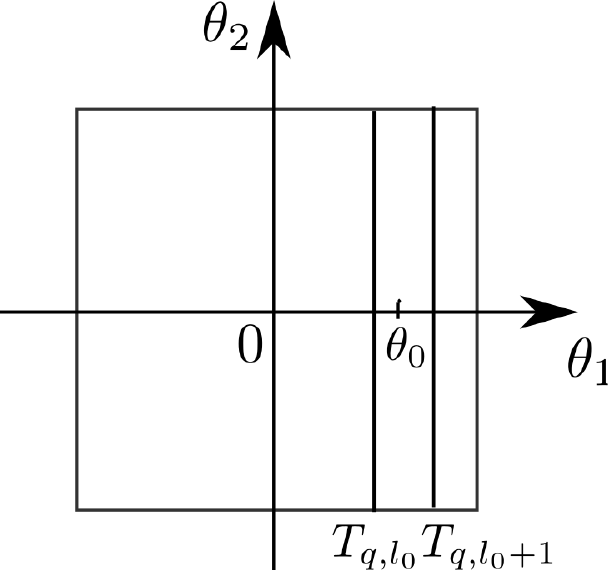}
\caption{Case $q_2=0, q_1>1$. Point $(\theta_0,0)$ lies between two lines $t_{q,l_0}$ and $t_{q,l_0+1}$}
\label{F:lemma1_10}
\end{figure}

For all such $k$ we have $2k\pi/q_1>\theta_0$. 
Since function $F_1$ is non-decreasing on $[0,\pi]$ for fixed $\theta_2$, we have
$$F_1(\theta_0,0)=-1/3\leq F_1(2k\pi/q_1,0)\leq -1/3 \text{ i.e. } F_1(2k\pi/q_1,0)=-1/3$$
and 
$$F_1(2(l_0+1)\pi/q_1,\pi)\leq F_1(2k\pi/q_1,\pi) \text{ for all }  k\in[l_0+1,[q_1/2]].$$ 
Thus, 
\begin{equation}
\bigcup_{l_0< k\leq \left[\frac{q_1}{2}\right]} F_1(T_{q,k})
=\left[F_1\left(\frac{2(l_0+1)\pi}{q_1},\pi\right),-\frac{1}{3}\right].
\label{E:after_k0}
\end{equation}

Combining (\ref{E:before_k0}) and (\ref{E:after_k0}) we obtain that
$$F_1(V_q)=[-1,F_1(2l_0\pi/q_1,0)]\cup [F_1(2(l_0+1)\pi/q_1,\pi),-1/3].$$
If $q_1=1$ then $V_q=T_{q,0}$, which leads to $F_1(V_q)=F_1(T_{q,0})=[-1,-2/3]$.

Now let us study the case when $q_2=2k_0$ for some nonzero integer $k_0$ (see Fig. \ref{F:Fig10}).
Note that $(0,-\pi)\in T_{q,k_0}$, thus $-1=F_1(0,-\pi)\in F_1(T_{q,k_0})$. 

If $T_{q,k_0}$ intersects with $\theta_1=\pi$ at some point $(\pi,\theta_2^0)$, then $F_1(\pi,\theta_2^0)=-1/3$.
As a consequence, $F_1(V_q)\supset F_1(T_{q,k_0})=[-1,-1/3]$ or $F_1(V_q)=[-1,-1/3]$.
\begin{figure}[ht!]
\includegraphics[width=84mm]{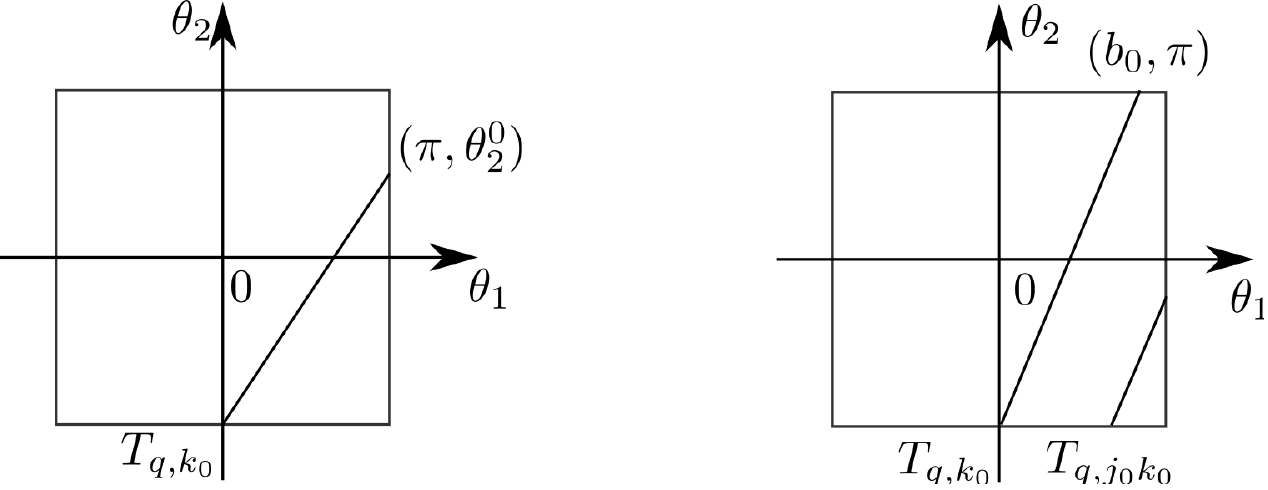}\caption{Different situations occur for even $q_2=2k_0$. 
$T_{q,k_0}$ intersects with $\theta_1=\pi$ on the left and with $\theta_2=\pi$ on the right}\label{F:Fig10}
\end{figure} 
If $T_{q,k_0}$ intersects with $\theta_2=\pi$ at $(b_0,\pi)$,
then $t_{q,jk_0}$ intersects with $\theta_2=-\pi$ and $\theta_2=\pi$ at $(a_j,-\pi)$ and $(b_j,\pi)$ correspondingly with $a_j=(j-1)b_0/2$ and $b_j=(j+1)b_0/2$.
There exists a smallest $j_0\in \mathbb{N}, j_0\geq 2$ such that $T_{q,j_0 k_0}$ intersects with both $\theta_2=-\pi$ at $(a_{j_0},-\pi)$ and $\theta_1=\pi$ at $(\pi,\theta_2^0)$.
Note that $0<a_j<a_{j+1}=b_{j-1}$ for all $j\geq 1$, thus according to the second remark
$$F_1(a_j,-\pi)<F_1(a_{j+1},-\pi)=F_1(a_{j+1},\pi)=F_1(b_{j-1},\pi).$$
\\
Then we have 

$\displaystyle [-1,-1/3]\subset$
\\
$\displaystyle \subset [-1,F_1(b_0,\pi)]\cup(\bigcup_{1\leq j<j_0}[F_1(b_{j-1},\pi),F_1(b_j,\pi)])\cup [F_1(b_{j_0-1},\pi),-1/3]$
\\
$\displaystyle \subset [-1,F_1(b_0,\pi)]\cup(\bigcup_{1\leq j<j_0}[F_1(a_j,-\pi),F_1(b_j,\pi)])\cup [F_1(a_{j_0},-\pi),-1/3]$
\\
$\displaystyle \subset \bigcup_{0\leq j\leq j_0} F_1(T_{q,jk_0})\subset F_1(V_q)$
\\
i.e. $F_1(V_q)=[-1,-1/3]$. 

We will now study the last case when $q_2$ is odd, namely $q_2=2k_0+1, k_0\in\mathbb{N}$.
Since $q_2$ is odd, $V_q$ contains neither $(0,\pi)$ nor $(0,-\pi)$, the minimum of $F_1(V_q)$ is some $a\in(-1,-2/3]$ (since $F_1(0,0)=-2/3$).
One might expect to be able to find $a$ on $T_{q,k_0}$ or $T_{q,k_0+1}$ which are closest to $(0,-\pi)$.

If $q_1=0$, $V_q=\{[-\pi,\pi]\times\{-2k\pi/q_2\}, k=0,\ldots,[q_2/2]\}$. 
Since $F_1(\theta_1,\theta_2)=F_1(-\theta_1,-\theta_2)$, $F_1(V_q)=\cup_{0\leq k\leq [q_2/2]}F_1([0,\pi]\times \{2k\pi/q_2\})$. 
Moreover, function $F_1$ is nondecreasing on $[0,\pi]$ for fixed $\theta_2$ and so 
$$F_1(V_q)=\cup_{0\leq k\leq [q_2/2]}[F_1(0,2k\pi/q_2),F_1(\pi,2k\pi/q_2)].$$
For $\theta_1=\pi$, $F_1\equiv -1/3$. 
Besides, function $F_1(0,\theta_2)$ is decreasing on $[0,\pi]$, we have 
$$F_1(V_q)=\bigcup_{0\leq k\leq [q_2/2]}\big[F_1\big(0,\frac{2k\pi}{q_2}\big),-1/3\big]= \big[F_1\big(0,\frac{2[q_2/2]\pi}{q_2}\big),-1/3\big].$$
Now we consider the case when $q_1\neq 0$, let $k_1\in\mathbb{N}$ such that function $F_1$ attain its minimum $a$ on $T_{q,k_1}$. 
If $T_{q,k_1}$ intersects with $\theta_1=\pi$ then we have $[a,-1/3]\subset F_1(T_{q,k_1})\subset F_1(V_q)$ i.e. $F_1(V_q)=[a,-1/3]$.
Otherwise let $k_2\in\mathbb{N}$ such that $T_{q,k_2}$ is the most left segment from $V_q$ having nonempty intersection with $\theta_1=\pi$ (see Fig. \ref{F:p2odd}).
\begin{figure}[ht!]
\includegraphics[width=39 mm]{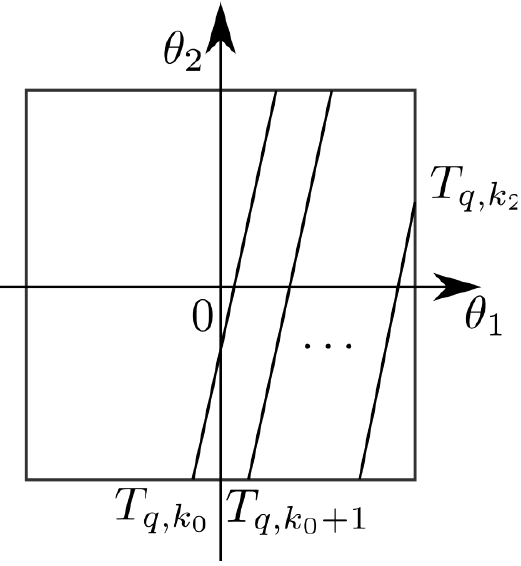}
\caption{$q_2$ is odd and $q_1$ is nonzero.
$T_{q,k_2}$ is the most left segment from $V_q$ having nonempty intersection with $\theta_1=\pi$}\label{F:p2odd}
\end{figure}
Let 
$(a_k,-\pi)$ and $(b_k,\pi)$ be intersection of $t_{q,k}$ with $\theta_2=-\pi$ and $\theta_2=\pi$ correspondingly. 
Then 
$$\displaystyle a_{k_0+1}=\frac{\pi}{q_1}\leq \frac{(2k_0+1)\pi}{q_1}=b_0\leq b_{k_1},$$
and 
$$a_k=\frac{(2k-q_2)\pi}{q_1}, b_k=\frac{(q_2+2k)\pi}{q_1}$$
which implies that $0< a_{k}\leq b_{k-1}<b_{k}$ for $k\geq k_0+1$.

Since function $F_1$ is non-decreasing on $[0,\pi]\times \{\pi\}$, we have
$$F_1(a_{k_0+1},-\pi)=F_1(a_{k_0+1},\pi)\leq F_1(b_{k_1},\pi)$$
and
$$F_1(a_{k},-\pi)= F_1(a_{k},\pi)\leq F_1(b_{k-1},\pi)\leq F_1(b_{k}, \pi)$$
which makes $[F_1(b_{k-1},\pi),F_1(b_{k}, \pi)]\subset [F_1(a_{k},-\pi),F_1(b_{k}, \pi)]\subset F_1(T_{q,k})$ for $k\geq k_0+1$.
\\
Thus 
\\$\displaystyle [a, -1/3]=[a,F_1(a_{k_0+1},-\pi)]\cup[F_1(a_{k_0+1},-\pi),F_1(b_{k_0+1},\pi)]\cup$
\\
$\displaystyle \cup(\cup_{k_0+1< k < k_2} [F_1(b_{k-1},\pi), F_1(b_{k}, \pi)]) \cup [F_1(b_{k_2-1},\pi),-1/3]\subset$
\\ 
$\displaystyle \subset [a,F_1(b_{k_1},\pi)]\cup F_1(T_{q,k_0+1})\cup (\cup_{k_0+1<k<k_2} F_1(T_{q,k}))\cup [F_1(a_{k_2}),-1/3]\subset$
\\
$\displaystyle \subset F_1(T_{q,k_1})\cup (\cup_{k_0+1\leq k \leq k_2} F_1(T_{q,k}))\subset F_1(V_q)$
\\ i.e. $F_1(V_q)=[a,-1/3]$.

Minimum value of $F_1(V_q)$ does not exceed $-2/3=F_1(0,0)$ in all cases. 

\subsection{Proof of Lemma \ref{L:F2}}
\begin{figure}[ht!]
\includegraphics[width=84mm]{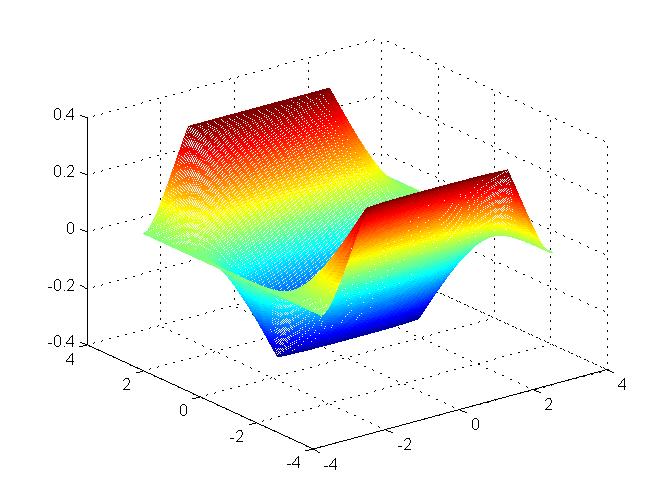}
\caption{Graph of function $F_2$}
\end{figure}
If $q_2\neq 0$ and $\displaystyle \frac{q_1}{q_2}\geq \frac{\pi}{\theta_0}$ or $q_2=0$, 
then $T_{q,0}$ intersects with $\theta_2=\pi$ at $(\theta_1^0, \pi)$ for some $\theta_1^0\in [0,\theta_0]$.
Thus $\displaystyle F_2(T_{q,0})=[-1/3,1/3]$, i.e. $F_2(V_q)=[-1/3,1/3]$.

If $q_2\neq 0$ and $1< q_1/q_2< \pi/\theta_0$, then $T_{q,0}$ and $\theta_2=\pi$ intersect at $(\theta_1^0,\pi)$ for some $\theta_1^0\in(\theta_0,\pi)$. 
Thus, $[-1/3,0]\subset F_2(T_{q,0})$.
Also since $q_1\geq 2, q_2\geq 1$, which is obtained from inequality $1<q_1/q_2$, we have
\begin{equation}
k_0:=\left[\frac{q_1\theta_0+q_2\pi}{2\pi}\right]\geq \left[\frac{2\theta_0+\pi}{2\pi}\right]\geq 1.
\label{E:lemma3_1}
\end{equation}
Two lines $t_{q,k_0}$ and $\theta_2=-\pi$ intersect at $\displaystyle (\theta_1^0, -\pi), \theta_1^0=\frac{-q_2\pi+2k_0\pi}{q_1}$.

From (\ref{E:lemma3_1}) one have 
$$k_0\leq \frac{q_1\theta_0+q_2\pi}{2\pi}< k_0+1 $$
therefore 
$$-\theta_0<\theta_0-\frac{2\pi}{2}\leq \theta_0-\frac{2\pi}{q_1}=\frac{-q_2\pi+q_1\theta_0+q_2\pi-2\pi}{q_1}<\frac{-q_2\pi+2k_0\pi}{q_1},$$
and
$$\frac{-q_2\pi+2k_0\pi}{q_1}\leq \frac{-q_2\pi+q_1\theta_0+q_2\pi}{q_1}=\theta_0,$$
i.e. $\theta_1^0\in [-\theta_0,\theta_0]$ and so $(\theta_1^0,-\pi)\in V_q$. 
\begin{figure}[ht!]
\includegraphics[width=84mm]{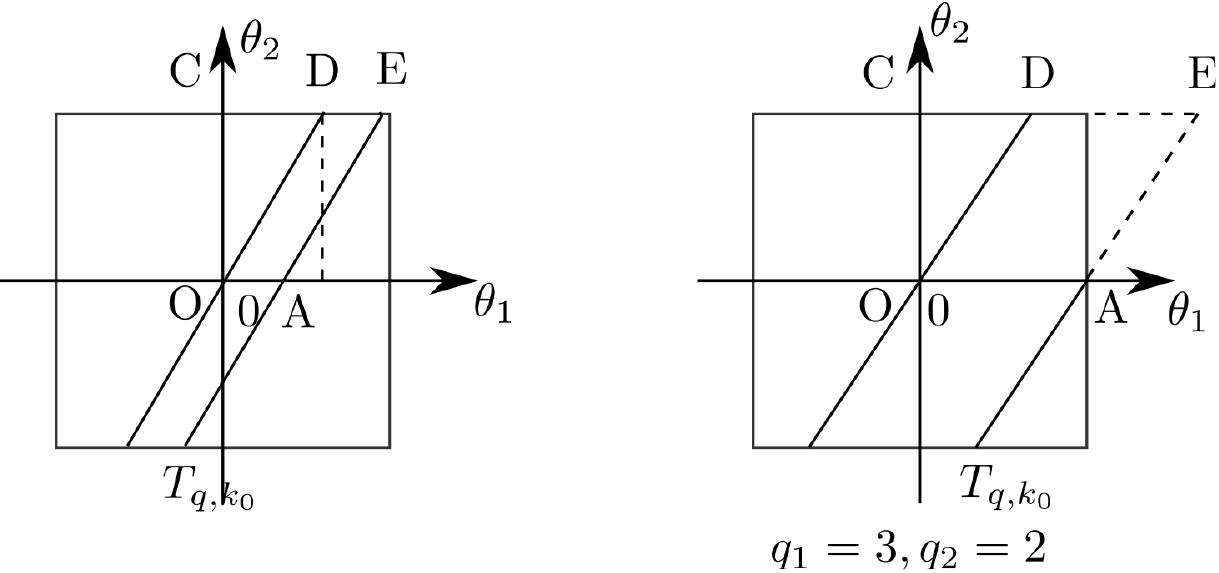}
\caption{$q_2\neq 0$ and $1<q_1/q_2<\pi/\theta_0$. On the left $0A\leq \theta_0$ while on the right is one case when $OA>\theta_0$. 
Here $k_0=[(q_1\theta_0+q_2\pi)/2\pi]$}\label{F:Fig6}
\end{figure}
Let $C=(0,\pi), O=(0,0)$, $A, E$ are intersection points of $t_{q,k_0}$ with two lines $\theta_2=0$ and $\theta_2=\pi$ correspondingly as shown in the Fig. \ref{F:Fig6}.

If $OA\leq\theta_0$, then $A\in V_q$ and $F_2(A)=-1/3$. Thus $[-1/3,1/3]\subset F_2(T_{q,k_0})$ i.e. $F_2(V_q)=[-1/3,1/3]$.
\begin{figure}[ht!]
\includegraphics[width=39mm]{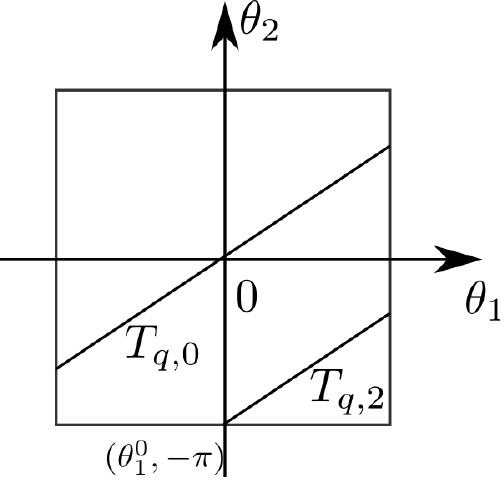}\caption{$T_{q,0}$ and $T_{q,k_0}$ in case $q_1/q_2\leq 1$ and $q_1>1$ (here $p=(2,3), k_0=2$)} \label{F:Fig7} 
\end{figure}
In case $OA>\theta_0$, we have $CE=CD+DE=CD+OA>\theta_0+\theta_0>\pi$ which means the point $E$ lies outside Brillouin zone $B$.
In this case, $T_{q,k_0}$ intersects with $\theta_1=\pi$ at $(\pi,\theta_2^0)$ for some $\theta_2^0\in [-\pi,0]$, 
thus $[0,1/3]=[F_2(\pi,\theta_2^0),F_2(\theta_1^0,-\pi)]\subset F_2(T_{q,k_0})$.
We also have that $F_2(V_q)=[-1/3,1/3]$ because
$$[-1/3,1/3]=[-1/3,0]\cup [0,1/3] \subset F_2(T_{q,0})\cup F_2(T_{q,k_0})\subset F_2(V_q).$$
Now we consider the case when $q_2\neq 0,  q_1/q_2\leq 1$ and $q_1>1$. Since
$$\frac{q_2\pi+q_1\theta_0}{2\pi}-\frac{q_2\pi-q_1\theta_0}{2\pi}=\frac{q_1\theta_0}{\pi}\geq\frac{2\theta_0}{\pi}>1,$$
we can always choose an integer $k_0$ such that
$$0<\frac{q_2\pi-q_1\theta_0}{2\pi}\leq k_0\leq \frac{q_2\pi+q_1\theta_0}{2\pi}.$$

For chosen $k_0$, two lines $t_{q,k_0}$ and $\theta_2=-\pi$ intersect at the point $(\theta_1^0,-\pi)$ for some $\theta_1^0\in [-\theta_0,\theta_0]$.

$F_2(V_q)=[-1/3,1/3]$ because
$$[-1/3,1/3]=[-1/3,0]\cup [0,1/3] \subset F_2(T_{q,0})\cup F_2(T_{q,k_0})\subset F_2(V_q).$$

The only case left to be considered is when $q_1\leq 1$ and $q_2\neq 0$ .

If $q_1=0$ and $q_2$ is even, the interval $T_{q,[q_2/2]}$ concides with $[-\pi,\pi]\times\{\pi\}$, so 
$$[-1/3,1/3]=[-1/3,0]\cup [0,1/3] \subset F_2(T_{q,0})\cup F_2(T_{q,[q_2/2]})\subset F_2(V_q).$$

If $q_1=0$ and $q_2$ is odd, namely $q_2=2k_0+1, k_0\in\mathbb{N}$, $V_q$ does not intersect with $\theta_2=\pm\pi$, thus $a:=\max F_3(V_q)\in[0,1/3)$.

Let $k\in\mathbb{N}$ such that $a\in F_2(T_{q,k})$.  
One would expect that $k=k_0$ which makes $t_{q,k}$ be closest to $\theta_2=-\pi$. 
We have 
$[0,a]\subset F_2(T_{q,k})$,
therefore 
$$[-1/3,0]\cup[0,a]\subset F_2(T_{q,0})\cup F_2(T_{q,k})\subset F_2(V_q) \text{  i.e.  } F_2(V_q)=[-1/3,a].$$

\begin{figure}[ht!]
\includegraphics[width=84mm]{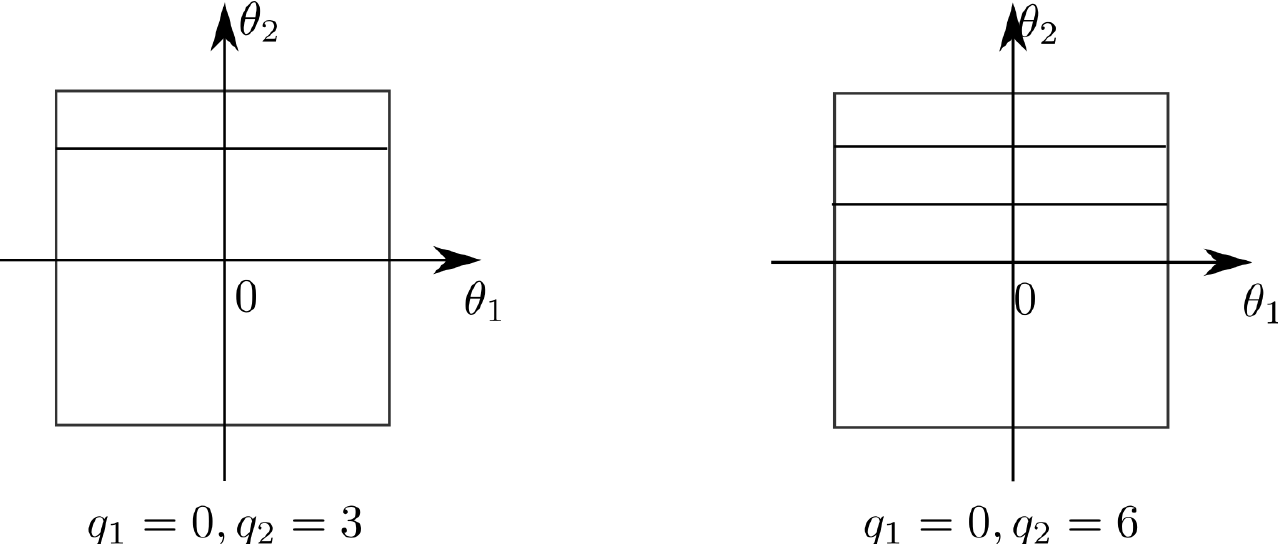}\caption{Different cases occur when $q_1=0$. On the left $q_2$ is odd and on the right $q_2$ is even}\label{F:Fig8}
\end{figure}

If $q_1=1$ and $q_2\neq 0$ even, i.e. $q_2=2k_0$ for some $k_0\in\mathbb{N}, k_0>0$,
we have $F_2(T_{q,k_0})=[0,1/3]$.
As a consequence,
$$[-1/3,1/3]=[-1/3,0]\cup[0,1/3]\subset F_2(T_{q,0})\cup F_2(T_{q,k_0})\subset F_2(V_q),$$
which means that $F_2(V_q)=[-1/3,1/3].$

If $q_1=1$ and $q_2$ is odd, namely, $q_2=2k_0+1, k_0\in\mathbb{N}$, $V_q$ again does not contains any point from $[-\theta_0,\theta_0]\times \{\pm\pi\}$, so 
$a:=\max F_2(V_q)\in[0,1/3).$
Let $k\in\mathbb{N}$ such that $a\in F_2(T_{q,k})$, since $(2k_0+1)(\pm\pi)\neq \theta_1 -2k\pi$ for $\theta_1\neq \pm\pi$, $T_{q,k}$ intersects with $\theta_1=\pm\pi$, and so
$[0,a]\subset F_2(V_q)$.

As a consequence, $$[-1/3,a]=[-1/3,0]\cup[0,a]\subset F_2(T_{q,0})\cup F_2(T_{q,k})\subset F_2(V_q).$$
Thus $F_2(V_q)=[-1/3,a]$ for $a\in[0,1/3)$. 

\subsection{Proof of Lemma \ref{L:F3}}

In this proof we will need the following remarks:

i. For fixed $\theta_2^0$ function $F_3(\theta_1,\theta_2^0)$ is non-increasing on $[0,\pi]$. 

ii. For fixed $\theta_1^0$ function $F_3(\theta_1^0,\theta_2)$ is decreasing on $[0,\pi]$.

The above remarks can be obtained using the same argument as in lemma \ref{L:F1}. 
\begin{figure}[ht!]
\includegraphics[width=84mm]{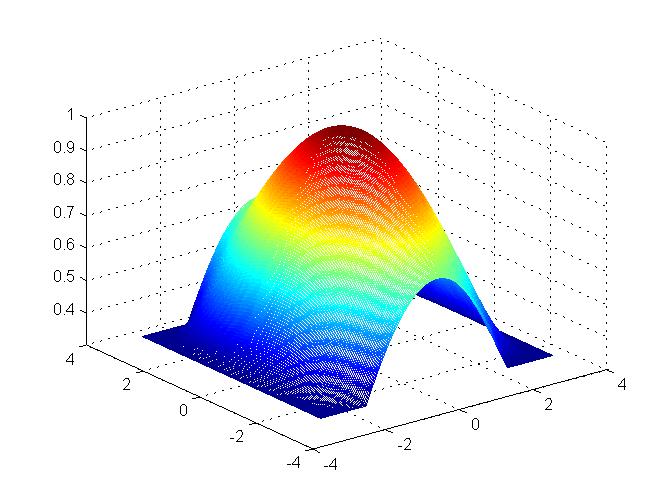}
\caption{Graph of function $F_3$}
\end{figure}

Notice that function $F_3$ always attains its maximum at $(0,0)$.

If $q_2=0$ then
$$V_q=\{\{2k\pi/q_1\}\times[-\pi,\pi], 0\leq k\leq [q_1/2]\}.$$
Since $F_3(\theta_1,\theta_2)=F_3(\theta_1,-\theta_2)$ we have 
$$F_3(V_q)=F_3(\{\{2k\pi/q_1\}\times[0,\pi], 0\leq k\leq [q_1/2]\}).$$
Thus according to the last remark
$$F_3(V_q)=\cup_{0\leq k\leq [q_1/2]} [F_3(2k\pi/q_1,\pi),F_3(2k\pi/q_1,0)].$$
For all $k$ such that $0\leq k\leq l_0=[q_1\theta_0/2\pi]$ we have
\begin{enumerate}[i.]
  \item $F_3(T_{q,0})=[2/3,1]$,
	\item $F_3(2l_0\pi/q_1,\pi)\leq F_3(2k\pi/q_1,\pi)<2/3=F_3(0,\pi)$ according to the first remark,
	\item $2/3=F_3(\theta_0,0)\leq F_3(2k\pi/q_1,0)$.
\end{enumerate}
Therefore 
\begin{equation}
\cup_{0\leq k\leq l_0} F_3(T_{q,k})=[F_3(2l_0\pi/q_1,\pi),1].
\label{E:before_l0}
\end{equation}

From lemma \ref{L:F1} we know that when $q_1>1$ there must be some $k\in(l_0,[q_1/2]]$. For these $k$ we have
$\theta_0<2k\pi/q_1\leq \pi$ which makes $1/3=F_3(\pi,\pi)\leq F_3(2k\pi/q_1,\pi)\leq F_3(\theta_0,\pi)=1/3$ or $F_3(2k\pi/q_1,\pi)=1/3$. Besides
$F_3(2k\pi/q_1,0)\leq F_3(2(l_0+1)\pi/q_1,0)$ from the first remark.
Thus, 
\begin{equation}
\bigcup_{l_0< k\leq \left[\frac{q_1}{2}\right]} \left[\frac{1}{3}, F_3\left(\frac{2k\pi}{q_1},0\right)\right]
=\left[\frac{1}{3},F_3\left(\frac{2(l_0+1)\pi}{q_1},0\right)\right].
\label{E:after_l0}
\end{equation}
Combining equations (\ref{E:before_l0}) and (\ref{E:after_l0}) we obtain that 
$$F_3(V_q)=[1/3,F_3(2(l_0+1)\pi/q_1,0)]\cup [F_3(2l_0\pi/q_1,\pi),1].$$
In case $q_1=1$, $V_q=T_{q,0}$ and so $F_3(V_q)=[2/3,1]$. 

Now we consider the case when $q_2\neq 0$.

If the slope $q_1/q_2$ of $t_{q,k}$ is less or equal than $\pi/\theta_0$,
then the interval $T_{q,0}$ intersects either with the line $\theta_1=\pi$ at $(\pi,\theta_2^0)$ for some $\theta_2\in[0,\pi]$ or with the line $\theta_2=\pi$ at $(\theta_1^0,\pi)$ for some $\theta_1^0\geq \theta_0$
(see Fig. \ref{F:Fig1}).
Since $F_3(\theta_1^0,\pi)=F_3(\pi,\theta_2^0)=1/3$, the range of $F_3$ restricted to $T_{q,0}$ is $[1/3,1]$, and thus $F_3(V_p)=[1/3,1]$.
\begin{figure}[ht!]
\includegraphics[width=84mm]{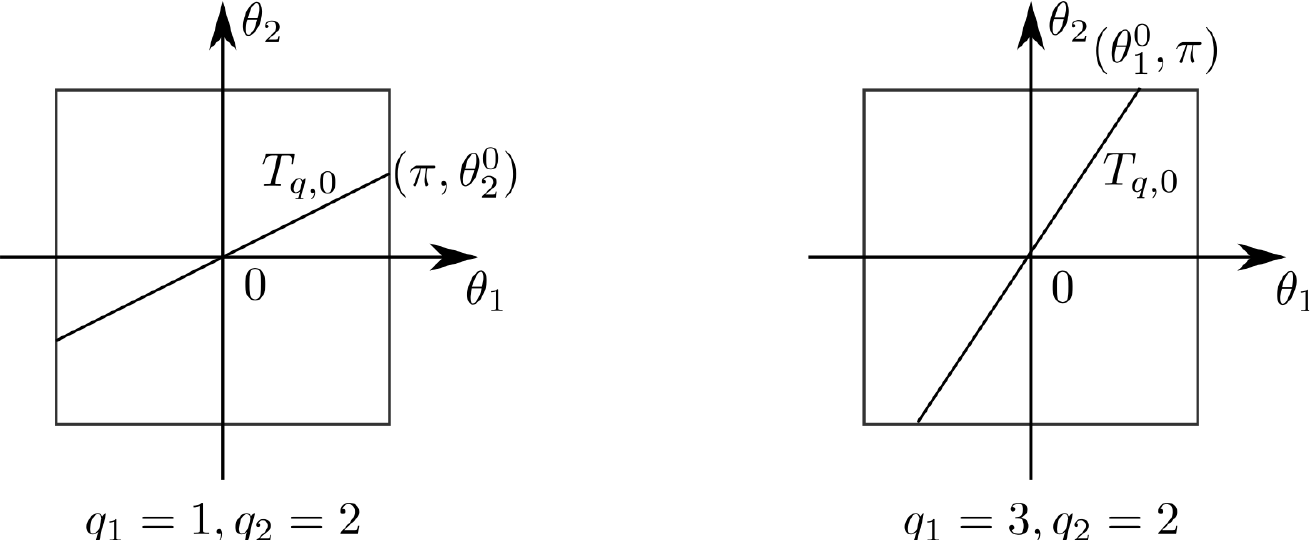}\caption{$q_2$ is nonzero and $q_1/q_2\leq\pi/\theta_0$.
On the left $T_{q,0}$ intersects with $\theta_1=\pi$ and on the right it intersects with $\theta_2=\pi$ at $(\theta_1^0,\pi)$ where $\theta_1^0>\theta_0$} \label{F:Fig1}
\end{figure}

If $q_1/q_2$ is larger than $\pi/\theta_0$, the interval $T_{q,0}$ intersects with the line $\theta_2=\pi$ at $(\theta_1^0,\pi)$
where $0 < \theta_1^0 < \theta_0$. Thus $F_3(T_{q,0})$, and as a consequence, $F_3(V_q)$ contains $[a,1]$ for $a=F_3(\theta_1^0,\pi)\in(F_3(\theta_0,\pi),F_3(0,\pi))=(1/3,2/3)$ according to the first remark.

In case $q_1\geq 4$ and $q_2>1$, let $k_0=\left[\frac{q_1\theta_0}{2\pi}\right]$, then
$$k_0\leq \frac{q_1\theta_0}{2\pi} <k_0+1.$$

Two lines $t_{q,k_0}$ and $\theta_2=0$ intersect at $(2k_0\pi/q_1,0)$.
Since $$0\leq\frac{2k_0\pi}{q_1}\leq\theta_0,$$
according to the first remark, we have
$$ b:= F_3\left(\frac{2k_0\pi}{q_1},0\right) > \frac{2}{3}.$$
Two lines $t_{q,k_0}$ and $\theta_2=\pi$ intersect at the point $((q_2+2k_0)\pi/q_1,\pi)$.
Since
$$\theta_0<\frac{2(1+k_0)\pi}{q_1}=\frac{(2k_0+2)\pi}{q_1}\leq\frac{(q_2+2k_0)\pi}{q_1},$$
the interval $T_{q,k_0}$ intersects either with the line $\theta_2=\pi$ at $(\theta_1^0,\pi)$ for some $\theta_0<\theta_1^0<\pi$
or with the line $\theta_1=\pi$ at $(\pi,\theta_2^0)$ for some $\theta_2^0\in [0,\pi]$ (see Fig. \ref{F:Fig2}).
\begin{figure}[ht!]
\includegraphics[width=84mm]{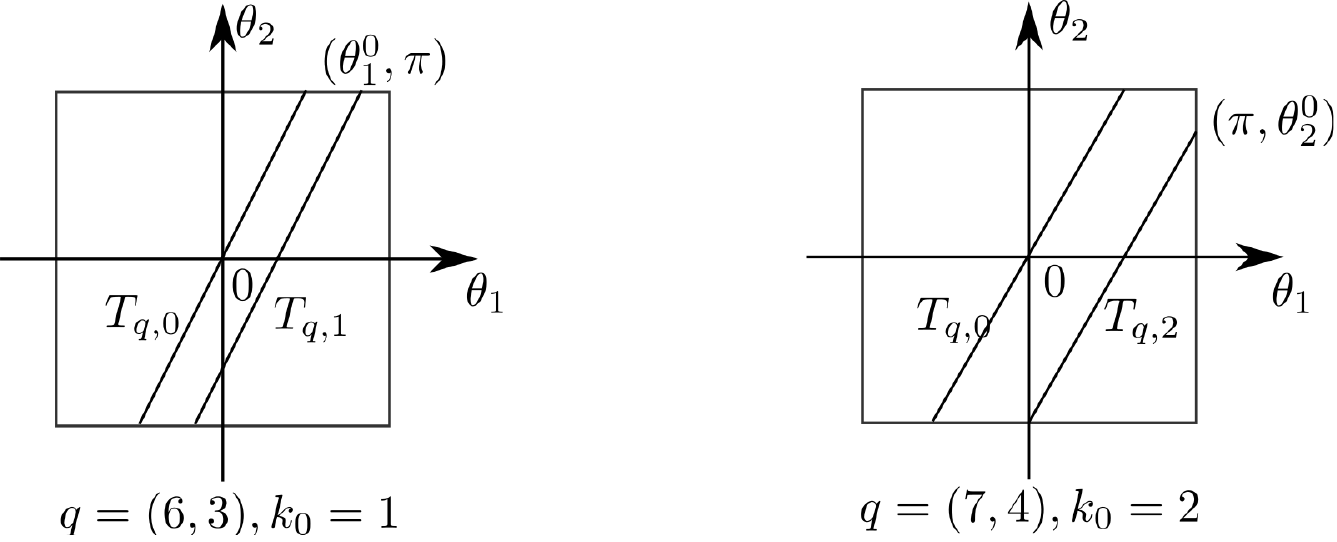}\caption{$T_{q,0}$ and $T_{q,k_0}$ in different cases for $q_1/q_2>\pi/\theta_0$ and $q_1\geq 4, q_2>1$.
On the left $T_{q,k_0}$ intersects with $\theta_2=\pi$ at $(\theta_1^0,\pi)$ for $\theta_1^0>\theta_0$ while on the right $T_{q,k_0}$ intersects with $\theta_1=\pi$ }\label{F:Fig2}
\end{figure}
\\
In both cases,
$$[1/3,1]= [1/3,b]\cup[a,1] \subset F_3(T_{q,k_0})\cup F_3(T_{q,0})\subset F_3(V_q)$$
i.e. $F_3(V_q)=[1/3,1].$

Now one need to consider the case when $q_1\geq 4$ and $q_2=1$ or $q_1<4$ and $q_1/q_2>\pi/\theta_0$.
Notice that for $q_1<4$, since $q_1/q_2>\pi/\theta_0$, $q_2$ can be 1 and $q_1>1$.
Thus it is enough to study the range of function $F_3$ restricted to $V_q$ for $q=(q_1,1)$ with $q_1>1$.

For each $k\in \mathbb{N}, k>0$, the line $t_{q,k}$ intersects with $\theta_2=-\pi$ at $((2k-1)\pi/q_1,-\pi)$
and with $\theta_2=\pi$ at $((2k+1)\pi/q_1,\pi)$.
When both intersection points are in $V_q$, we have $[b_k,a_k]\subset F_3(T_{q,k})$, where
$$a_k:=F_3\left(\frac{(2k-1)\pi}{q_1},-\pi\right), b_k:=F_3\left(\frac{(2k+1)\pi}{q_1},\pi\right).$$

On the other hand
$$F_3\left(\frac{(2k+1)\pi}{q_1},\pi\right)=F_3\left(\frac{(2k+1)\pi}{q_1},-\pi\right), \text{ i.e. } a_{k+1}=b_k.$$
\begin{figure}[ht!]
\includegraphics[width=84mm]{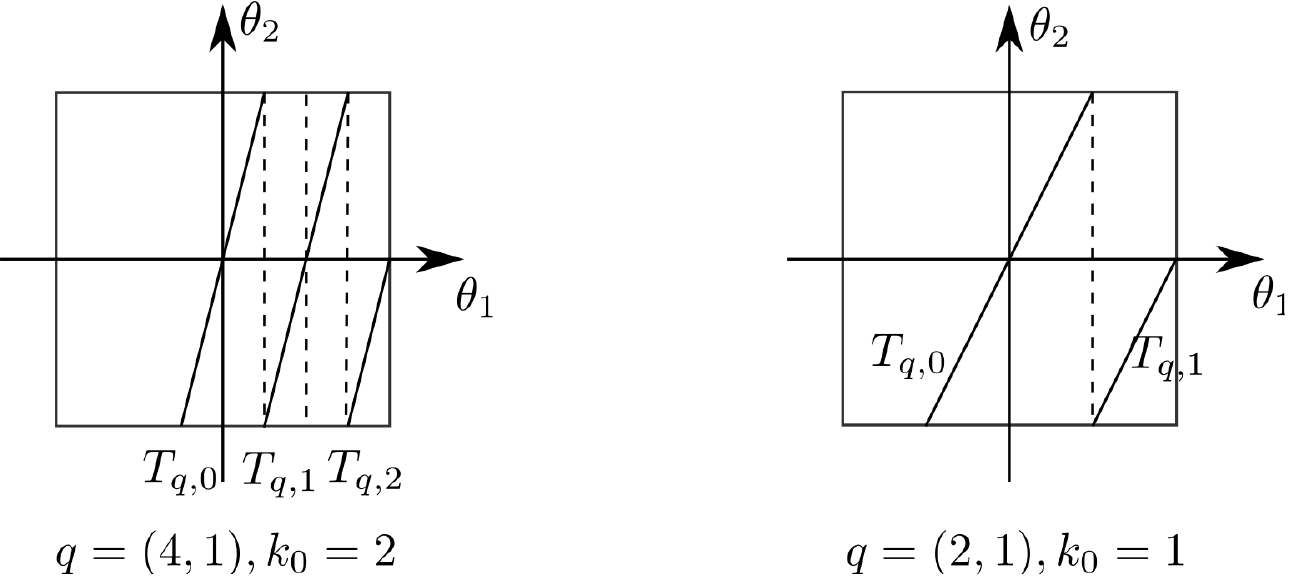}\caption{$T_{q,k}$ for $k=\overline{0,k_0}$ in different cases when $q_1>1, q_2=1$}\label{F:Fig3}
\end{figure}
\\
Besides $[b_0,1]\subset F_3(T_{q,0})$ and there must be some $k_0$ such that $t_{q,k_0}$ intersects with $\theta_1=\pi$ inside $V_q$,
i.e. $[1/3,a_{[q_1/2]}]\subset F_3(T_{q,k_0})$.
Therefore,
$$[1/3,1] = \bigcup_{0\leq k< k_0}[b_k,a_k]\cup [b_{0},1] \cup[1/3,a_{k_0}] \subset\bigcup_{0\leq k \leq k_0}F_3(T_{q,k}),$$
so $F_3(V_q)=[1/3,1].$

\section*{Acknowledgments}
The author is grateful to P. Kuchment for insightful discussion and comments and to the anonymous referees for their subtantial remarks.

\end{document}